\documentclass[12pt,journal,draftclsnofoot,onecolumn,english]{IEEEtran}
\usepackage{cite}
\usepackage{amsmath,amssymb,amsfonts}
\usepackage{graphicx}
\usepackage{subcaption}
\usepackage{cleveref}
\usepackage{color, xcolor}
\usepackage{multirow}
\usepackage{pstricks}

\usepackage{amsthm}
\usepackage{booktabs}
\usepackage{multirow}
\usepackage{algorithm}
\usepackage{algpseudocode}
\usepackage{caption}
\newtheorem{theorem}{Theorem}
\newtheorem{lemma}{Lemma}
\newtheorem{definition}{Definition}
\newtheorem{corollary}{Corollary}

\newtheorem{example}{Example}
\newtheorem{remark}{Remark}
\newtheorem{construction}{Construction}

\newcommand{\tabcaption}{\def\@captype{table}\caption}
\newcommand{\tabincell}[2]{\begin{tabular}{@{}#1@{}}#2\end{tabular}}
\begin{document}
\title{A New Construction Structure on Coded Caching with Linear Subpacketization:  Non-Half-Sum Disjoint Packing}
	\author{Minquan Cheng,~\IEEEmembership{Member,~IEEE,} Huimei Wei,~\IEEEmembership{Student Member,~IEEE,} 		
		Kai Wan,~\IEEEmembership{Member,~IEEE,}  and Giuseppe Caire,~\IEEEmembership{Fellow,~IEEE}
\thanks{
A short version of this paper   was submitted to  the 2025 IEEE International Symposium on Information Theory. 
}
\thanks{M. Cheng and H. Wei are with Guangxi Key Lab of Multi-source Information Mining $\&$ Security, Guangxi Normal University,
	Guilin 541004, China (e-mail:  chengqinshi@hotmail.com, scarlett\_w\_edu@163.com). The work of M.~Cheng was in part supported by 2022GXNSA035087, 2024GXNSFGA010001,  and the BAGUI Young Scholar Program of Guangxi.}
\thanks{K.~Wan is with the School of Electronic Information and Communications,
	Huazhong University of Science and Technology, 430074  Wuhan, China,  (e-mail: kai\_wan@hust.edu.cn). The work of K.~Wan was partially funded by the   National Natural
Science Foundation of China (NSFC-12141107),  the Key Research and Development Program of Wuhan under Grant 2024050702030100, and Wuhan ``Chen Guang''
Pragram under Grant 2024040801020211}
\thanks{G. Caire is with the Electrical Engineering and Computer Science Department, Technische Universit\"{a}t Berlin,10587 Berlin, Germany (e-mail: caire@tu-berlin.de). The work of G.~Caire was partially funded by the Gottfried Wilhelm Leibniz-Preis 2021 of the German Science Foundation (DFG).}
}
\date{}
\maketitle
\begin{abstract}
Coded caching is a promising technique to effectively reduce peak traffic by using local caches and the multicast gains generated by these local caches. Coded caching schemes have been widely investigated following the seminal work of Maddah-Ali and Niesen. Explicit coding constructions have been proposed for a variety of network topologies with information-theoretically optimal or near-optimal transmission load $R$. An important parameter in these constructions is the subpacketization $F$, i.e., the number of subpackets that each content file needs to be divided into. In particular, the original scheme of Maddah-Ali and Niesen as well as several other variants require $F$ to grow exponentially with the number of users $K$. In practice, files have finite size and too large $F$ yields impractically small subpackets. Therefore, it is important to design coded caching schemes with $F$ and $R$ as small as possible. At present, the few known schemes with subpacketization linear in $K$ achieve large load. In this paper, we consider the linear scaling regime $F = O(K)$ and design schemes with unprecedented transmission load $R$. 
 Specifically, we first introduce a new combinatorial structure called non-half-sum disjoint packing (NHSDP) which can be used to generate a coded caching scheme with $K=O(F)$. A class of new schemes is then obtained by constructing NHSDP. Theoretical and numerical analyses demonstrate that (i) in comparison to existing schemes with linear subpacketization, 
the proposed scheme achieves a lower load; (ii) the proposed scheme also attains a lower load than some existing schemes with polynomial subpacketization in some cases; and (iii) the proposed scheme achieves load values comparable to those of existing schemes with exponential subpacketization in some cases. Furthermore, the newly introduced concept of NHSDP is closely related to classical combinatorial structures, including cyclic difference packings (CDP), non-three-term arithmetic progressions (NTAP), and perfect hash families (PHF). These relationships underscore the significance of NHSDP as an important combinatorial structure in the field of combinatorial design. 
 
\end{abstract}

\section{INTRODUCTION}
The rise of internet connected devices in recent years has significantly increased network traffic, fueled by activities such as multimedia streaming, web browsing, and social networking. In addition, the high temporal variability of this traffic leads to congestion during peak periods and inefficient use of network resources during off-peak times. Caching is a promising technique to reduce peak traffic by taking advantage of memories distributed across the network to duplicate content during off-peak times. Conventional uncoded caching techniques focus on predicting the user demands to make appropriate prefetching decisions, thus realizing a ``local caching gain", which scales with the amount of local memory \cite{BGW}. In their seminal paper \cite{MN}, Maddah-Ali and Niesen (MN) demonstrated that, in addition to a local caching gain, coded caching can also attain a ``global caching gain". This global caching gain  scales with the global amount of memory in the network, since each transmission of the MN scheme can serve multiple users  simultaneously.

 Following \cite{MN}, the basic coded caching problem consists of the following setup: a central server has $N$ equal-size files and it is connected to $K$ users via a shared (broadcast) link of normalized capacity, which is able to transmit one file per unit of time. Each user has a local cache memory of capacity equivalent to $M$ files. An $F$-division $(K,M,N)$ coded caching scheme consists of two phases: the placement phase during the off-peak hours and the  delivery phase during the peak hours. In the placement phase, the server divides each file into $F$ equal-size packets and places some of these packets into the users’ caches without knowledge of the users' future demands. If packets are directly placed in the user's caches, the scheme is referred to as ``uncoded placement". If the server places more general functions of the packets in the users’ caches, the scheme is referred to as ``coded placement". The parameter $F$ is referred to as the subpacketization order of the scheme (in brief, subpacketization).

In the delivery phase, each user requests an arbitrary file. According to the users' demands and the cache contents, the server broadcasts some coded packets such that each user can recover its desired file. The normalized amount of transmission for the worst-case over all possible demands is called the transmission load (or load) $R$. In particular, a scheme achieving load $R$ can serve all users’ demands by transmitting at most $R$ times the size of one file. Following \cite{MN}, we define the caching gain as the reduction factor in the load with respect to sending 
 sequentially all demanded files. For example, in the case of $K\leq N$ and distinct demands, a system without caching needs to send all $K$ files through the shared link. A system with traditional uncoded caching (e.g., using prefix caching, where a fraction $M/N$ of each file is cached) requires the transmission of $(1-M/N)K$ files. A coded caching system achieving load $R$ achieves a ``coded caching gain" $K(1 - M/N)/R\geq 1$ over the system without caches and the system with conventional uncoded caching.

The first well-known coded caching scheme was proposed by Maddah-Ali and Niesen in \cite{MN}, which utilized a combinatorial design in the placement phase and  linear coding in the delivery phase. This scheme is referred  to as the MN scheme. The authors in \cite{WTP2016,JCLC,YMA2018,WTP2020} separately proved the minimality of the load of the MN scheme under the uncoded  placement and for the case $K \leq N$. Using the placement strategy of the MN scheme, the authors in \cite{YMA2018} proposed an exact optimal caching scheme, which strictly improves upon the current state of the art by exploiting commonalities among user demands under uncoded placement.

However, the subpacketization $F={K\choose KM/N}$ of the MN scheme increases exponentially with the number of users $K$. This would become infeasible for practical implementation when $K$ is large. So it is meaningful to design a coded caching scheme with low subpacketization. The grouping method, which is widely regarded as the most effective in reducing subpacketization, was proposed in \cite{SJTLD,CJWY}. Nevertheless, the load of the scheme generated by grouping method increases fast. Recently, the authors in \cite{YCTC} proposed an interesting combinatorial structure called placement delivery array (PDA) to study coded caching schemes with low subpacketization. They also showed that the MN scheme is equivalent to a special PDA which is called MN PDA. A large amount of articles have used the PDA approach to devise low subpacketizaton coded caching schemes (e.g., \cite{CJWY, YCTC,CJYT,CJTY,CWZW,WCWG,WCWL,WCCLS,ZCW,LC,PKB,WCLC,AST}).

There are other  characterizations of coded caching schemes, such as using linear block codes \cite{TR}, the special $(6,3)$-free hypergraphs \cite{SZG}, the $(r,t)$ Ruzsa-Szem\'{e}redi graphs\cite{STD}, the strong edge coloring of bipartite graphs\cite{YTCC}, cross resolvable designs\cite{KMR}, the projective geometry \cite{CKSM}, and combinatorial designs \cite{ASK}, which are listed in Table~\ref{tab-Schemes}. 
In this table, we do not include the schemes from \cite{STD} and the first scheme in \cite{XXGL}, since we focus on explicit constructions, while theirs only focus on the existence of constructions and the user number approximates an infinite integer.
\begin{table}[http!]
\renewcommand{\arraystretch}{2}
\setlength\tabcolsep{1pt} 
\centering
\caption{The existing coded caching schemes where $K,k,t,m,H,a,z,r \in \mathbb{Z}^{+}$, $\left[k \atop t \right]_q=\frac{(q^{k}-1)\dots(q^{k-t+1}-1)}{(q^{t}-1)\dots(q-1)}$, $\left \langle K \right \rangle_{t}=K \mod t$.}
\label{tab-Schemes}
\begin{tabular}{|c|c|c|c|c|c|}\hline
Schemes & Number of Users & Memory ratio  & Load    & Subpacketization &Constraint\\ \hline

MN Scheme  \cite{MN}&$K$               & $\frac{t}{K}$& $\frac{K-t}{t+1}$& ${K\choose t}$         &   \\ \hline
\tabincell{c}{WCLC  \cite{WCLC}}&$\binom{m}{z}k^z$               & $1-\left(\frac{k-t}{k}\right)^{z}$& $(\frac{k-t}{\lfloor\frac{k-1}{k-t}\rfloor})^z$ & ${\lfloor\frac{k-1}{k-t}\rfloor}^zk^{m-1}$         &\tabincell{c}{$1\leq t < k$,\\[-0.5cm] $1\leq z \leq m$}    \\ \hline


YTCC  \cite{YTCC}&$\binom{H}{a}$& $1-\frac{\binom{a}{r}\binom{H-a}{z-r}}{\binom{H}{z}}$
&\tabincell{c}{$\frac{\binom{H}{a+z-2r}}{\binom{H}{z}}\cdot$ \\[-0.5cm] $\min\{\binom{H-a-z+2r}{a-r},$\\[-0.5cm] $\binom{a+z-2r}{a-r}\}$}
& $\binom{H}{z}$&\tabincell{c}{$r< a < H,$\\[-0.5cm] $r<z<H,$\\[-0.5cm]$a+z \leq H+r$}          \\ \hline

WCCLS  \cite{WCCLS}  &$q^m$ &$1-\frac{\binom{m}{w}(q-1)^w}{q^m}$ &$\frac{\binom{m}{w}(q-1)^w}{q^{m-w}}$&$q^m$&$m,w \in \mathbb{Z}^{+},m<w$
\\ \hline

\tabincell{c}{CKSM  1 \cite{CKSM}}
&\tabincell{c}{$\frac{1}{t!}q^{\frac{t(t-1)}{2}}\cdot$\\[-0.5cm]$\prod \limits_{i=0}^{t-1}\left[k-i \atop 1 \right]_q$}
&\tabincell{c}{$1-q^{mt}\cdot$ \\[-0.5cm] $\prod \limits_{i=0}^{m-1}\frac{\left[k-t-i \atop 1\right]_q}{\left[k-i \atop 1\right]_q}$}& \tabincell{c}{$\frac{m!q^{mt}}{(m+t)!}q^{\frac{t(t-1)}{2}}\cdot$\\[-0.5cm] $\prod \limits_{i=0}^{t-1}\left[k-m-i \atop 1\right]_q$}&
\tabincell{c}{$\frac{1}{m!}q^{\frac{m(m-1)}{2}}\cdot$\\[-0.5cm]$\prod \limits_{i=0}^{m-1}\left[k-i \atop 1 \right]_q$}&
\tabincell{c}{$m+t \leq k$,\\[-0.5cm] prime power}\\ \cline{1-5}

\tabincell{c}{CKSM  2 \cite{CKSM}}&
$\left[ k \atop t\right]_q$&$1-\frac{\left[k-t \atop m \right]_q}{\left[k \atop m+t\right]_q}$                  &$\frac{\left[k \atop m\right]_q}{\left[k \atop m+t\right]_q}$      &$\left[k \atop m+t\right]_q$&$2\leq q$ \\ \hline

\tabincell{c}{ASK  1 \cite{ASK}}&
$q^2+q+1 $&$\frac{q^2}{q^2+q+1}$&$1$&$q^2+q+1$& prime power     \\  \cline{1-5}
\tabincell{c}{ASK  2   \cite{ASK} }
&$q^2$       & $\frac{q-1}{q}$              &$\frac{q}{q+1}$      &$q^2+q$   & $2\leq q$         \\ \hline

\tabincell{c}{ZCW  \cite{ZCW} } &$2^{m}$ &$1-\frac{\binom{m}{\omega}}{\sum_{i=0}^{\omega}\binom{m}{i}}$  &$\frac{\binom{m}{\omega}2^{m-\omega}}{\sum_{i=0}^{\omega}\binom{m}{i}}$ & $\sum_{i=0}^{\omega}\binom{m}{i}$  &  $\omega <  m$ \\ \hline

& & &$\frac{(K-t)(K-t+1)}{2K}$ &$K$ &\tabincell{c}{ $(K-t+1)|K$ \\[-0.5cm]or $K-t=1$} \\ \cline{4-6}
WCWL  \cite{WCWL} & \tabincell{c}{$K$}  &$\frac{t}{K}$  &$\frac{K-t}{2\lfloor \frac{K}{K-t+1}\rfloor+1}$ &$\left(2\lfloor \frac{K}{K-t+1}\rfloor+1\right)K$ &\tabincell{c}{$\left \langle K \right \rangle_{K-t+1}=K-t$ } \\ \cline{4-6}
&   & &$\frac{K-t}{2\left\lfloor \frac{K}{K-t+1}\right\rfloor}$ &$2\left\lfloor \frac{K}{K-t+1}\right\rfloor K$ &\tabincell{c}{ $\mbox{otherwise}$ } \\ \hline
XXGL  \cite{XXGL} &$K$ &$\frac{K-2}{K}$ &$\frac{K-1}{K}$ &$K$ & $K \in \mathbb{N}^{+}$ \\ \hline
AST  \cite{AST} &$2^{r}k$  &$1-\frac{r+1}{2^{r}}+\frac{r}{2^{r}k}$ &$\frac{k(r+1)-r}{2^{r}}$  &$2^{r}k=K$   &$r,k \in \mathbb{N}^{+}$ \\ \hline
MR  \cite{MR}  &$K$  &$\frac{t}{K}$  &\small$\left\lceil\frac{K(K-t)}{2+\left \lfloor\frac{t}{K-t+1}\right \rfloor+\left \lfloor\frac{t-1}{K-t+1}\right\rfloor}\right\rceil \cdot \frac{1}{K}$ &$K$  & \\ \hline
\end{tabular}
\end{table}

From Table~\ref{tab-Schemes}, we notice that the schemes in \cite{CJWY,MN,YCTC,CJYT,TR,SZG,WCWG,WCLC,CWZW}  have a flexible number of users, large memory regimes, and a small load, but a large subpacketization which increases exponentially (or sub-exponentially) with the number of users; the schemes in \cite{YTCC,ZCW,WCWL,SS,AST,WCCLS,CKSM,ASK,XXGL,MR} have a low subpacketization which increases polynomially or linearly with the number of users. However, the schemes in \cite{CKSM,YTCC,WCCLS} are restricted to the special number of users (i.e., combinations, powers, or products of combinations and powers) and special memory ratio (i.e., the ratios of these combinations, powers, or products of combinations and powers); the schemes in \cite{ZCW,ASK,XXGL} have a large memory ratio close to $1$; the schemes in \cite{AST,CKSM} have  memory ratios close to $0$ or $1$. The schemes in \cite{MR,WCWL} have a large subpacketization increasing linearly with the number of users for the flexible number of users and large memory regimes. It is worth noting that the schemes in \cite{MR,WCWL} use the same placement strategy, namely, the so-called consecutive cyclic uncoded placement. The authors showed that under the consecutive cyclic uncoded placement and clique-covering delivery, the maximum coded caching gain is $2\lfloor\frac{K}{K-KM/N+1}\rfloor+1$. Clearly, when $M/N$ is small such that $M/N\leq 1/2$,  linear coded caching schemes with one-short delivery can achieve a coded caching gain of at most $3$.

\subsection{Contribution}
In this paper, we focus on constructing coded caching schemes with linear subpacketization. When $K$ is odd, we introduce a new combinatorial concept called non-half-sum disjoint packing (NHSDP), which can realize a coded caching scheme with linear subpacketization. Compared to the existing characterizing methods, the greatest advantage of NHSDP is that it integrates the placement strategy and the transmission strategy into a single condition, i.e., the second condition of NHSDP in Definition \ref{def-NHSDP}. The main result can be summarized as follows. 

\begin{itemize}
\item  When $K$ is odd, we transform the coded caching construction  to a new combinatorial structure named non-half-sum disjoint packing (NHSDP). Given a $(v,g,b)$ NHSDP for any odd positive integers $v$, $g$, and $b$, we can obtain a $(K=v,M,N)$ coded caching scheme with the memory ratio $\frac{M}{N}=1-\frac{bg}{v}$,  subpacketization $F=K$, and transmission load $R=b$. If $K$ is even, we can add one virtual user to the system such that the effective number of users is $K'=K+1$, and then construct a scheme based on the $(K’+1, g, b)$ NHSDP.  
\item By constructing an NHSDP, we can obtain a $(K=v,M,N)$ coded caching scheme with the memory ratio $M/N=1-\frac{2^n\prod_{i=1}^{n}m_i}{v}$, the coded caching gain $g=2^n$, and the transmission load $R=\prod_{i=1}^{n}m_i$ for any odd integer  $v\geq 2\sum_{i=1}^{n-1}(m_i\prod_{j=i+1}^{n}(1+2m_{i}))+ m_{n}+1$  where $n$ and $m_1$, $m_2$, $\ldots$, $m_n$ are any positive integers. In particular, when $m_1=\cdots=m_n=\lfloor\frac{v^{1/n}-1}{2}\rfloor$, we obtain a $(K=v,M,N)$ coded caching scheme with the memory ratio $\frac{M}{N}=1-\frac{2^n\lfloor\frac{v^{1/n}-1}{2}\rfloor^n}{v}$, subpacketization $F=K$, coded caching gain $g=2^n$, and  transmission load $R={\left\lfloor \frac{v^{1/n}-1}{2} \right\rfloor}^n$.

\item  Theoretical and numerical comparisons show that our proposed scheme achieves a lower load compared to the existing schemes in \cite{AST,WCWL,ZCW,XXGL} with linear subpacketization;  the proposed scheme achieves a lower load than some existing schemes in \cite{CKSM,YTCC} with polynomial subpacketization  under some system parameters; compared to some existing schemes in  \cite{MN,WCLC} with exponential subpacketization, our scheme incurs only a slight load increase for some system parameters. 
\item The new concept of NHSDP has a close relationship to other classic combinatorial structures, such as cyclic difference packing (CDP) \cite{dm/Yin98}, non-three-term arithmetic progressions (NTAP) \cite{brown1982density},  and perfect hash family (PHF) \cite{NM1996}, etc. Specifically, a CDP can be used to construct a $(v,g,1)$ NHSDP. A $(v,g,1)$ NHSDP is equivalent to an NTAP set over $[v]$ and it can be used to construct a $(3:gv,v,3)$ PHF.  In cryptography \cite{NM1996}, combinatorics \cite{JS2011}, and coding theory \cite{CG2016},  we aim to construct a subset of $\mathbb{Z}_n$ satisfying NTAP, with the maximum cardinality, and a $(3:m,v,3)$ PHF with the maximum value of $m$. 
When $v=3^n$, we can obtain a class of  NTAPs using our new NHSDP whose size is larger than the state-of-the-art  achievable bound on the NTAP size in \cite{elsholtz2024improving} when $n\leq 52$.
The PHF derived from the proposed NHSDP has more columns compared to  the first quadrics PHF\cite{Hara-PHF}. Furthermore, the number of columns in the resulting PHF  is  close to that of the Hermitian PHF presented in \cite{Hara-PHF}.

\end{itemize}

\subsection{Organizations and notations}
The rest of this paper is organized as follows. In Section~\ref{sec-perlimin}, a coded caching system, the PDA, and their relationship are introduced. In Section~\ref{sec-NHSDP}, we introduce the non-half-sum disjoint packing (NHSDP) and show that it can be used to generate a PDA. In Section~\ref{sec-Construct-NHSDP}, we propose a new class of PDAs by constructing an NHSDP. In  Section~\ref{sec-perf-ana}, we provide the performance analysis of the proposed scheme, and we expand the application of NHSDP in Section~\ref{sec-expanding}. Finally, we conclude this paper in Section~\ref{sec-conclu}.  

\textit{Notation:} In this paper, we will use the following notations. Let bold capital letter, bold lowercase letter, and curlicue letter  denote array, vector, and set respectively; let $|\mathcal{A}|$ denote the cardinality of the set $\mathcal{A}$; define  $[a]=\{1,2,\ldots,a\}$ and $[a:b]$ is the set $\{a,a+1,\dots,b-1,b\}$; $\lfloor a\rfloor $ denotes the largest integer not greater than $a$. $a \nmid b$ indicates that $a$ does not divide $b$. We define that $\left[k \atop t \right]_q=\frac{(q^{k}-1)\dots(q^{k-t+1}-1)}{(q^{t}-1)\dots(q-1)}$, $\left \langle K \right \rangle_{t}=K \mod t$;  $\mathbb{Z}_v$ is the ring of integer residues modulo $v$.

\section{PRELIMINARIES}\label{sec-perlimin}
In this section, we will review a  coded caching system, placement delivery array, and their relationship.
\subsection{Coded caching system}\label{sec:system mode}
A $(K,M,N)$ coded caching system  illustrated in Figure~\ref{system model} contains a server storing  $N$ equal-sized files $\mathcal{W}=\{W_n\ |\ n\in[N]\}$, $K$ users each of which can cache at most $M$ files where $0\leq M \leq N$. The server connects to the users over an error-free shared-link. An $F$-division $(K,M,N)$ coded caching scheme operates in two phases:
\begin{figure}[h]
\centering
\includegraphics[height=6cm,width=8cm]{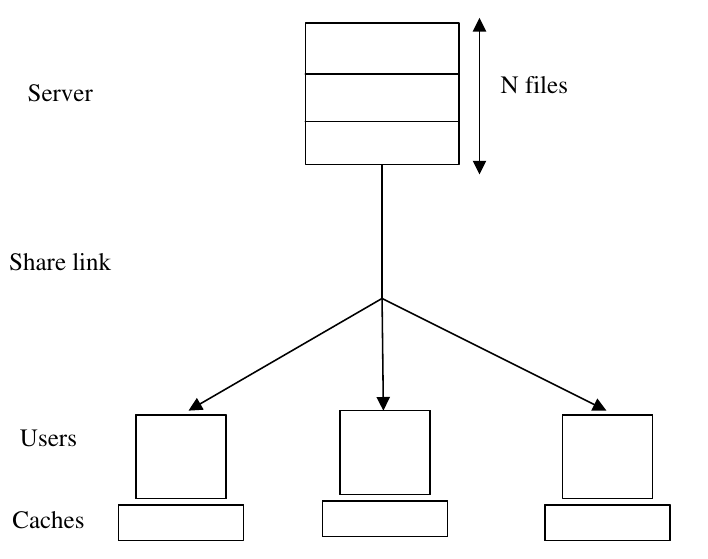}
\caption{$(K,M,N)$ caching system}
\label{system model}
\end{figure}
\begin{itemize}
\item {\bf Placement phase:} The server divides each file into $F$ equal-sized packets, i.e., $W_n=\{W_{n,j}\ |\ j\in[F]\}$ where $n\in[N]$, and places some packets into each user's cache without knowing any information about  future demands. Let $\mathcal{Z}_{k}$ be the packets cached by user $k$. This placement strategy is called uncoded placement, otherwise it is called coded placement. 

\item {\bf Delivery phase:} Each user requests one file from $\mathcal{W}$ randomly. Denote the requested file number by ${\bf d}$=$(d_{1},d_{2},\ldots,d_{K})$, i.e., user $k$ requests file $W_{d_{k}}$, where $k \in \mathcal{K},d_{k} \in [N]$. After receiving the request vector ${\bf d}$, the sever broadcasts XOR of coded packets with a size of at most $R_{\bf d}F$ to the users such that each user is able to decode its requested file.
\end{itemize}
In this paper, we focus on the normalized amount of transmission for the worst-case over all possible demands which is defined as follows.
\begin{equation}
\label{eq:def of load}
R=\max\{R_{\bf d}\ |\ {\bf d}\in[N]^K\}.
\end{equation}
The first well-known scheme was proposed in \cite{MN}, which is referred to as the MN scheme, has the minimum load under uncoded placement \cite{WTP2016,JCLC,YMA2018,WTP2020}, but the MN scheme has large subpacketization which increases exponentially with the number of users.

\subsection{Placement delivery array}
In order to study linear schemes with lower subpacketization, the authors in \cite{YCTC} proposed an interesting combinatorial structure called placement delivery array defined as follows.
\begin{definition}[\cite{YCTC}]\rm
\label{def-PDA}
For positive integers $K,F,Z$ and $S$, an $F \times K $ array $\mathbf{P}=(p_{j,k})$, where $j \in [F]$ and $k \in [K]$, composed of $"*"$ called star and $[S]$, is called a $(K,F,Z,S)$ placement delivery array (PDA) if the following conditions hold:
\begin{enumerate}
\item [C$1$.] Each column has exactly $Z$ stars.
\item [C$2$.] Each integer in $[S]$ occurs at least once.
\item [C$3$.] For any two distinct entries $p_{j_{1},k_{1}}$ and 
$p_{j_{2},k_{2}}$, $p_{j_{1},k_{1}} = p_{j_{2},k_{2}}=s$ if only if
\begin{enumerate}
\item [a.] $j_{1} \neq j_{2}, k_{1} \neq k_{2}$, i.e., they lie in distinct rows and distinct columns; 
\item [b.] $ p_{j_{1}, k_{2}}= p_{j_{2}, k_{1}}=*$, i.e., the corresponding $2 \times 2$ subarray formed by rows $j_{1}, j_{2}$ and columns  $k_{1}, k_{2}$ must be one of the following form,
\begin{align}\label{eq-form-pda}
\left(\begin{array}{cc}
s &$*$\\
* & s
\end{array}\right)\ \ \ \ \textrm{or}\ \ \ \ 
\left(\begin{array}{cc}
* & s\\
s &$*$
\end{array}\right).
\end{align}
\end{enumerate}
\end{enumerate}
\end{definition}%
A PDA is called $g$-regular if each integer occurs exactly $g$ times. For example, one can check that the following array $\mathbf{P}$ is a $2$-$(4,4,2,4)$ PDA.
\begin{eqnarray}\label{ex-pda}
\mathbf{P}=\left(
\begin{array}{ccccc}
*&1&*&4\\
1&*&2&*\\
*&2&*&3\\
4&*&3&*\\
\end{array}
\right).
\end{eqnarray} 
By Algorithm~\ref{alg-PDA}, for any given PDA, we can obtain the following result.
\begin{lemma}\rm(\cite{YCTC})\label{le-PDA}
Given a $(K,F,Z,S)$ PDA, there exists an $F$-division coded caching scheme for the $(K,M,N)$ coded caching system with memory ratio $\frac{M}{N}=\frac{Z}{F}$, subpacketization $F$, and load $R=\frac{S}{F}$.
\hfill $\square$
\end{lemma}In plain words, the relationship between a PDA and the corresponding coded caching scheme is as follows. Each $k$-th column of the PDA describes the cache configuration of the corresponding $k$-th user. Each file is split into $F$ packets. User $k$ caches the packets corresponding to the $*$ symbol in the $k$-th PDA column. Thus, the amount of data cached by user $k$ is $Z$ (the number of $*$ symbols) times the (normalized) size of one packet, which is $1/F$. This yields a memory ratio $M/N = Z/F$. Then, each integer $s \in [S]$ corresponds to the transmission of one linearly coded packet. At transmission step $s$, the severs takes all the packets of the demanded files marked by entries equal to $s$ in $\mathbf{P}$, and transmits the XOR. By condition C2 above, each XOR is formed by packets needed by (at least one) user, and present in the caches of all other users whose requests participate in the XOR. Therefore, each user obtains either $1$ or $0$ packet from each XOR. It takes at most $S$ rounds to satisfy all users' demands. In particular, there is at least one configuration of the user demands that requires $S$ rounds. The amount of transmitted data (normalized to the size of one file) is therefore $R = S/F$. 

\begin{algorithm}
\caption{Caching scheme based on PDA in \cite{YCTC}}\label{alg-PDA}
\begin{algorithmic}[1]
\Procedure {Placement}{$\mathbf{P}$, $\mathcal{W}$}
\State Split each file $W_n\in\mathcal{W}$ into $F$ packets, i.e., $W_{n}=\{W_{n,j}\ |\ j\in[F]\}$.
\For{$k\in [K]$}
\State $\mathcal{Z}_k\leftarrow\{W_{n,j}\ |\ p_{j,k}=*, n=[N],j\in [F]\}$
\EndFor
\EndProcedure
\Procedure{Delivery}{$\mathbf{P}, \mathcal{W},{\bf d}$}
\For{$s\in[S]$}
\State  Server sends $\bigoplus_{p_{j,k}=s,j\in[F],k\in[K]}W_{d_{k},j}$.
\EndFor
\EndProcedure
\end{algorithmic}
\end{algorithm} 
For example, using Algorithm~\ref{alg-PDA}, a $4$-division $(4,2,4)$ coded caching scheme based on $\mathbf{P}$ in \eqref{ex-pda} can be obtained as follows.
\begin{itemize}
\item {\bf Placement phase}: From Line $2$, we split each file into $4$ packets, i.e., $W_n=\{W_{n,j}\ |\ n\in [4],j\in[4]\}$. By Lines $3$-$5$, the caches for all of users are
\begin{align*}
\mathcal{Z}_1=\{W_{n,1},W_{n,3} \ | \ n \in [4]\},\ \ \  \mathcal{Z}_2=\{W_{n,2},W_{n,4} \ | \ n \in [4]\}, \\
\mathcal{Z}_3=\{W_{n,1},W_{n,3} \ |\  n \in [4]\},\ \ \
\mathcal{Z}_4=\{W_{n,2},W_{n,4} \ | \ n \in [4]\}.
\end{align*}
Clearly each user caches exactly $2\times 4=8$ packets, i.e., $M=2$ files. So we have $M/N=2/4=1/2$, i.e., $M/N=Z/F$.
\item {\bf Delivery phase}: Assume that the request vector is $\mathbf{d}=(1,2,3,4)$. By Lines $8$-$10$, the server transmits  $W_{1,2} \oplus W_{2,1}$ at the first time slot, $W_{2,3} \oplus W_{3,2}$ at the second time slot, $W_{3,4} \oplus W_{4,3}$ at the third time slot, and $ W_{1,4} \oplus W_{4,1}$ at the last time slot. It is not difficult to check that each user can decode its requested file. There are exactly $4$ coded signals transmitted by the server. From \eqref{eq:def of load} we have $R=4/4=1$, i.e., $R=S/F$.
\end{itemize}
There are many constructions based on  PDA~\cite{CJWY,CJYT,CJTY,CWZW,ZCJ,ZCW,YCTC,YTCC,SZG,MW}. The authors in \cite{YCTC} have shown that the MN scheme can also be represented by a PDA. Other constructions can also be  represented by appropriate PDAs such as the caching schemes based on  liner block code~\cite{TR}, projective space~\cite{CKSM}, and combinatorial designs~\cite{ASK}, etc.  We list their performances in Table~\ref{tab-Schemes}. In addition, given a PDA, we can obtain a PDA with any number of users by the following result.
\begin{lemma}[Grouping method\cite{CJWY}]\label{le-fundamental recursive}\rm
Given a $(K_1,F,Z,S)$ PDA, there exists a $(K,h_{1}F,h_{1}Z, hS)$ PDA for any $K>K_1$ where $h_1=\frac{K_{1}}{gcd(K_{1},K)}$ and $h=\frac{K}{gcd(K_{1},K)}$.
\end{lemma}

\section{Non-half-sum disjoint packing}\label{sec-NHSDP}
Recall that Condition C$3$-a) of Definition~\ref{def-PDA} is called the Latin property in combinatorial design theory. Naturally, we can construct a PDA with $K=F$ by first constructing a Latin square and then putting stars in some integer entries such that the resulting array satisfies conditions C$1$, C$2$, and C$3$. We then obtain the desired PDA with linear subpacketization. Specifically, we use the classical and simple construction of the Latin square $\mathbf{L}=(l_{f,k})_{f,k\in \mathbb{Z}_v}$ where the entry $l_{f,k}=f+k$ for each $f$, $k\in \mathbb{Z}_v$. Recall that $\mathbb{Z}_v$ is a positive integer ring, i.e., the arithmetic operations in the ring
$\mathbb{Z}_v$. When we cyclically replace the integer entries by stars, condition C$3$ of Definition~\ref{def-PDA} is transformed into the following novel combinatorial structure for the integer entries of the first row of the Latin square.
\begin{definition}[Non-half-sum disjoint packing, NHSDP]\label{def-NHSDP}\rm
For any positive odd integer $v$, a pair $(\mathbb{Z}_v,\mathfrak{D})$ where $\mathfrak{D}$ consists of $b$ $g$-subsets of $\mathbb{Z}_v$ is called $(v,g,b)$ non-half-sum disjoint packing if it satisfies the following conditions.
\begin{itemize}
\item  The intersection of any two different elements in $\mathfrak{D}$ is empty;  
\item  For  each $\mathcal{D}\in \mathfrak{D}$,  the half-sum of any two different elements in $\mathcal{D}$ (i.e., the sum of the two elements divided by $2$\footnote{The operations are in $\mathbb{Z}_v$ and since $v$ is odd, then $2$ has an inverse, that is, $1/2$ in $\mathbb{Z}_v$ is also an element of $\mathbb{Z}_v$}) does not appear in any block of $\mathfrak{D}$.\footnote{\label{foot:block}In an NHSDP $(\mathbb{Z}_v,\mathfrak{D})$, each element of $\mathfrak{D}$ is called block.}  
\end{itemize} \hfill $\square$
\end{definition}

 Let us take the following example with $v=15$ to further explain the concept of NHSDP. 
\begin{example}\rm
\label{exam-v-NHSDP}Consider the following $(v=15, g=4, b=2)$ NHSDP 
\begin{align}\label{eq-ex-d1}
\mathfrak{D}= \{\mathcal{D}_1=\{-1,1,-2,2\}=\{14,1,13,2\},\ \mathcal{D}_2=\{-4,4,-5,5\}=\{11,4,10,5\}\}.
\end{align}Clearly $\mathcal{D}_1\cap\mathcal{D}_2=\emptyset$, i.e., the first condition of Definition~\ref{def-NHSDP} holds. 
The half-sums of any two different elements in each block in $\mathfrak{D}$ are as follows, 
\begin{equation}
\begin{aligned}
& {\rm in}\  \mathcal{D}_1:  \frac{-1+1}{2}=0,\  
\frac{-1-2}{2}=6,\ 
\frac{-1+2}{2}=-7,\  
\frac{1-2}{2}=7, \ 
\frac{1+2}{2}=-6,\  
\frac{-2+2}{2}=0,\\
&{\rm in}\  \mathcal{D}_2: \frac{-4+4}{2}=0,\  
\frac{-4-5}{2}=3,\ 
\frac{-4+5}{2}=-7,\  
\frac{4-5}{2}=7, \ 
\frac{4+5}{2}=-3,\  
\frac{-5+5}{2}=0.
\label{eq:ex D1 D2}
\end{aligned}
\end{equation}We have $\{0,\pm 3,\pm 6, \pm 7\}\cap\mathcal{D}_1=\{0,\pm 3,\pm 6, \pm 7\}\cap\mathcal{D}_2=\emptyset$, i.e., the second condition of Definition~\ref{def-NHSDP} holds. So $(\mathbb{Z}_{15},\mathfrak{D})$ is a $(15,4,2)$ NHSDP.\hfill $\square$
\end{example}

Using an NHSDP, we can obtain a PDA with $F=K$ by the following novel construction. 
\begin{construction}\rm\label{cons-1}
Given a $(v,g,b)$ NHSDP $(\mathbb{Z}_v,\mathfrak{D})$, then a $v\times v$ array $\mathbf{P}=(p_{f,k})_{f,k\in \mathbb{Z}_v}$ is defined in the following way
\begin{equation}\label{eq-cons-1}
p_{f,k}=\begin{cases}
(f+k,i),& \mbox{if\ } k-f \in \mathcal{D}_i,\ \exists\ i\in[b];\\
\ \ \ \ \ *, &\mbox{otherwise}.
\end{cases}
\end{equation} \hfill $\square$
\end{construction}
Now let us take the NHSDP in Example~\ref{exam-v-NHSDP} to further illustrate Construction~\ref{cons-1}.
\begin{example}\rm
\label{exam-PDA-1}
When $v=15$, we have a $(15,4,2)$ NHSDP in Example~\ref{exam-v-NHSDP}.  By Construction~\ref{cons-1}, the following array can be obtained,

\begin{align}\label{eq-ex-P1}
\addtocounter{MaxMatrixCols}{10}
\setlength{\arraycolsep}{1.0pt}
\mathbf{P}={\small
\begin{pmatrix}
*&(1,1)&(2,1)&*&(4,2)&(5,2)&*&*&*&*&(10,2)&(11,2)&*&(13,1)&(14,1)\\
(1,1)&*&(3,1)&(4,1)&*&(6,2)&(7,2)&*&*&*&*&(12,2)&(13,2)&*&(0,1)\\
(2,1)&(3,1)&*&(5,1)&(6,1)&*&(8,2)&(9,2)&*&*&*&*&(14,2)&(0,2)&*\\
*&(4,1)&(5,1)&*&(7,1)&(8,1)&*&(10,2)&(11,2)&*&*&*&*&(1,2)&(2,2)\\
(4,2)&*&(6,1)&(7,1)&*&(9,1)&(10,1)&*&(12,2)&(13,2)&*&*&*&*&(3,2)\\
(5,2)&(6,2)&*&(8,1)&(9,1)&*&(11,1)&(12,1)&*&(14,2)&(0,2)&*&*&*&*\\
*&(7,2)&(8,2)&*&(10,1)&(11,1)&*&(13,1)&(14,1)&*&(1,2)&(2,2)&*&*&*\\
*&*&(9,2)&(10,2)&*&(12,1)&(13,1)&*&(0,1)&(1,1)&*&(3,2)&(4,2)&*&*\\
*&*&*&(11,2)&(12,2)&*&(14,1)&(0,1)&*&(2,1)&(3,1)&*&(5,2)&(6,2)&*\\
*&*&*&*&(13,2)&(14,2)&*&(1,1)&(2,1)&*&(4,1)&(5,1)&*&(7,2)&(8,2)\\
(10,2)&*&*&*&*&(0,2)&(1,2)&*&(3,1)&(4,1)&*&(6,1)&(7,1)&*&(9,2)\\
(11,2)&(12,2)&*&*&*&*&(2,2)&(3,2)&*&(5,1)&(6,1)&*&(8,1)&(9,1)&*\\
*&(13,2)&(14,2)&*&*&*&*&(4,2)&(5,2)&*&(7,1)&(8,1)&*&(10,1)&(11,1)\\
(13,1)&*&(0,2)&(1,2)&*&*&*&*&(6,2)&(7,2)&*&(9,1)&(10,1)&*&(12,1)\\
(14,1)&(0,1)&*&(2,2)&(3,2)&*&*&*&*&(8,2)&(9,2)&*&(11,1)&(12,1)&*
\end{pmatrix}}.
\end{align}
It is not difficult to check that each column of $\mathbf{P}$ has exactly $Z=v-bg=15-2\times4=7$ stars, there are exactly $S=30$ vectors  in  $\mathbf{P}$, and each vector occurs at most once in each row and each column. So the conditions C$1$, C$2$, and C$3$-a) of Definition~\ref{def-PDA} hold. Finally, let us consider condition C$3$-b). Let us first consider the entries $p_{1,0}=p_{0,1}=(1,1)$. We have $p_{1,1}=p_{0,0}=*$ which satisfies condition C$3$-b). The reason for this is as follows. When $f=1$, if $k=0$ we have $k-f=-1=14\in \mathcal{D}_1$ in \eqref{eq-ex-d1} and $1+0=1$, then we set $p_{1,0}=(1,1)$; if $k=1$, we have $k-f=0$, which does not appear in either $\mathcal{D}_1$ or $\mathcal{D}_2$. Thus, we set $p_{1,1}=*$. Similarly, we can check that all the vectors in $\mathbf{P}$ satisfy condition C$3$-b) of Definition~\ref{def-PDA}. Therefore, $\mathbf{P}$ is a $(15,15,7,30)$ PDA that can realize a coded caching scheme with the memory ratio $\frac{M}{N}=\frac{7}{15}$, the subpacketization $F=K=15$, and load $R=2$.
\hfill $\square$
\end{example} 

  In Example~\ref{exam-PDA-1}, each column of $\mathbf{P}$ has exactly $|\mathcal{D}_1|+|\mathcal{D}_2|=|\mathcal{D}_1\cup \mathcal{D}_2|=8$. This implies that  condition C$1$ of  Definition~\ref{def-PDA} is guaranteed  by the first condition of NHSDP. From \eqref{eq:ex D1 D2}, the half-sum set is  $\mathcal{H}=\{0,3,6,7,8,9,12\}$. We can see that in Example~\ref{exam-PDA-1}, the stars of $\mathbf{P}$ in \eqref{eq-ex-P1} are exactly the entries $p_{i,i+h}$ for each $i\in \mathbb{Z}_{15}$ and $h\in\mathcal{H}$. So the half-sums of all the blocks generates the star positions. This implies that condition C$3$-b) of Definition~\ref{def-PDA} is ensured by the second condition of NHSDP. Recall that our construction is based on the constructing a Latin square. So the condition C$3$-a) always holds.
  
   Then, for any parameters $v$, $g$ and $b$, if there exists a $(v,g,b)$ NHSDP by Construction~\ref{cons-1}, we can also obtain a PDA in the following theorem, whose proof is given in Appendix~\ref{sec:proof-th-1}.
\begin{theorem}[PDA via  NHSDP]\rm\label{th-1}
Given a $(v,g,b)$ NHSDP, we can obtain a $(v,v,v-bg,bv)$ PDA which realizes a $(K=v,M,N)$ coded caching scheme with memory ratio $\frac{M}{N}=1-\frac{bg}{v}$, coded caching gain $g$, subpacketization $F=v$, and  transmission load $R=b$.\hfill $\square$
\end{theorem}

\begin{remark}\rm
\label{re-even-case}  When $K$ is even, we can add a virtual user into the coded caching system, making the efficient number of users $K+1$ which can be solved by the  $(K+1,g,b)$ NHSDP in Theorem~\ref{th-1}. 

\end{remark}

\section{Constructions of PDAs via NHSDPs}
\label{sec-Construct-NHSDP}
By Theorem~\ref{th-1}, in order to obtain a coded caching scheme with linear subpacketization we should study NHSDPs. In this section we will propose a new construction of NHSDPs. 

We first introduce the main idea of our new construction of NHSDPs. By embedding integers into high-dimensional geometric spaces, we aim to transform the NHSDP into studying a geometric problem. Specifically, we will design a subset of $\mathbb{Z}_v$, say $\mathcal{X}$, to construct an NHSDP $(\mathbb{Z}_v,\mathfrak{D})$ such that each integer in $\mathfrak{D}$ can be uniquely represented by $\mathcal{X}$, ensuring that the first condition of Definition~\ref{def-NHSDP} holds.  Furthermore,  the half-sum generated by any two different integers in the same block of $\mathfrak{D}$ can also be represented by $\mathcal{X}$, and the coefficient sets for these two representations are disjoint. This approach can naturally distinguish the elements in $\mathfrak{D}$ from their half-sums perfectly which implies that  the second condition of Definition~\ref{def-NHSDP} is guaranteed. By finding an appropriate subset $\mathcal{X}$, we obtain the following main construction in this paper.
\begin{construction}\rm
\label{constr-2}
For any $n$ positive integers $m_1$, $m_2$, $\ldots$, $m_n$, let $\mathcal{A}:=[m_1]\times [m_2]\times \cdots\times[m_n]$. Define the following recursive function
	\begin{align}
		\label{eq-RF}
		\left\{
		\begin{array}{cc} 
			f(1)=m_1,\ \ \ \ \ \ \ \ \ \ \ \ \ \ \ \ \ \ \ \ \ \ \ \ &\ \text{if}\ i=1;\\
			f(i)=m_i\left(2\sum_{j=1}^{i-1}f(j)+1\right), &\ \text{if}\ i\geq 2,
		\end{array} 
		\right.
	\end{align}and $\mathcal{X}:=\{x_i=\frac{f(i)}{m_i} | i\in [n]\}$.  We can construct a family $\mathfrak{D}=\{\mathcal{D}_{\bf a}\ |\ {\bf a}=(a_1,a_2,\ldots,a_n)\in \mathcal{A}\}$ where $\mathcal{D}_{\bf a}$ is defined as  
\begin{align}
\label{eq-family}
\mathcal{D}_{\bf a}=\left\{\alpha_1 a_1x_1+\alpha_2 a_2x_2+\cdots+\alpha_n a_nx_n\ \Big|\  \alpha_i\in\{-1,1\},i\in[n]\right\},
\end{align}for each vector ${\bf a}=(a_1,a_2,\ldots,a_n)\in \mathcal{A}$.  By the above construction, $\mathfrak{D}$ has $m_1 \times m_2 \times \cdots \times m_n$ blocks, and each block $\mathcal{D}_{\bf a}$ has $2^n$ integers. 
\hfill $\square$ 
\end{construction}
 
Let us take the following example to illustrate Construction~\ref{constr-2}.
\begin{example}\rm\label{exam-n=3}
When $n=3$, $m_1=m_2=m_3=2$, we have $\mathcal{A}=[2]^3$. From \eqref{eq-RF}, when $i=1$, $2$, and $3$ we have 
\begin{align*}
&f(1)=m_1=2,\ \  f(2)=m_2(2f(1)+1)=2(4+1)=10,\\ 
&f(3)=m_3(2f(1)+2f(2)+1)=2(4+20+1)=50,
\end{align*}respectively. Then we have the subset 
\begin{align*}
\mathcal{X}=\left\{x_1=\frac{f_1}{m_1}=\frac{2}{2}=1,x_2=\frac{f(2)}{m_2}=\frac{10}{2}=5,x_3=\frac{f(3)}{m_3}=\frac{50}{2}=25\right\}.
\end{align*}
Now let us consider $\mathfrak{D}$ in Construction~\ref{constr-2}. When ${\bf a}=(1,1,1)$ and $(x_1,x_2,x_3)=(1,5,25)$, from \eqref{eq-family}, then we have 
\begin{align*}
\mathcal{D}_{(1,1,1)}=&\{\alpha_1 a_1x_1+\alpha_2 a_2x_2+\alpha_3 a_3x_3\ |\ \alpha_1,\alpha_2,\alpha_3\in\{-1,1\}\}\\
=&\{1 \cdot \alpha_1+5 \cdot \alpha_2+25 \cdot \alpha_3 \ |\ \alpha_1,\alpha_2,\alpha_3\in\{-1,1\} \}\\
=&\{31,21,29,19,-19,-29,-21,-31
\}.
\end{align*}For instance, the integer $31$ in $\mathcal{D}_{(1,1,1)}$ can be obtained by $1+5+25=31$. The number of all  possible $(\alpha_1, \alpha_2, \alpha_3)$ is $2^3=8$; thus the block $\mathcal{D}_{(1,1,1)}$ contains $8$ different integers. Similarly, we can obtain the following blocks.

\begin{equation}\label{eq-points}
	\begin{split}
		\mathfrak{D}=\{&\mathcal{D}_{(1,1,1)}=\{31,21,29,19,-19,-29,-21,-31
		\},\\
		&\mathcal{D}_{(2,1,1)}=\{32,22,28,18,-18,-28,-22,-32
		\},\\ 
		&\mathcal{D}_{(1,2,1)}=\{36,16,34,14,-14,-34,-16,-36
		\},\\
	&	\mathcal{D}_{(2,2,1)}=\{37,17,33,13,-13,-33,-17,-37
		\},\\ 
		&\mathcal{D}_{(1,1,2)}=\{56,46,54,44,-44,-54,-46,-56
		\},\\
	&	\mathcal{D}_{(2,1,2)}=\{57,47,53,43,-43,-53,-47,-57
		\},\\ 
		&\mathcal{D}_{(1,2,2)}=\{61,41,59,39,-39,-59,-41,-61
		\},\\
	&	\mathcal{D}_{(2,2,2)}=\{62,42,58,38,-38,-58,-42,-62
		\}\}.
\end{split}
\end{equation} 
By the above equation, $\mathfrak{D}$ has $m_1m_2m_3=8$ blocks and the intersection of any blocks is empty. The minimum and  maximum values of $\mathfrak{D}$ are $-62$ and $62$, respectively. To ensure that these two sets $[-62:0]$ and $[1:62]$ do  not have overlap in $\mathbb{Z}_v$, the value of $v$ must satisfy $v \geq 2\times 62+1=125$, since $0$ is also in $\mathbb{Z}_{v}$.  Next, let we verify that $(\mathbb{Z}_{125},\mathfrak{D})$ is an $(125,8,8)$ NHSDP. Now let we consider the half-sum of $31$ and $21$, i.e., $\frac{31+21}{2}=26$, which is not in any block of $\mathfrak{D}$. Similarly we can check  the half-sum of any two different integers in each block of $\mathfrak{D}$.

The intuition on why the proposed construction leads to a  NHSDP is summarized as follows.
\begin{itemize}
\item  In order to ensure that the first condition of NHSDP, i.e., the intersection of any two blocks is empty, each integer in $\mathfrak{D}$ is uniquely represented by a linear combination of all elements in $\mathcal{X}$, and the coefficients of $(x_1,x_2,x_3)$ are $(a_1\alpha_1, a_2\alpha_2, a_3\alpha_3)\in\mathcal{C}= \{\pm 1,\pm 2\}^3$, since $(a_1,a_2,a_3)\in [2]^3$ and $(\alpha_1, \alpha_2, \alpha_3) \in \{-1,+1\}^3$.  
By the above structure of linear combinations and by  the choice of $\{x_i:i\in [3]\}$,  each block $\mathcal{D}_{\mathbf{a}}$ where $\mathbf{a} \in [2]^3$ has $2^3=8$ integers and any two different blocks $\mathcal{D}_{\mathbf{a}_1}$  and $\mathcal{D}_{\mathbf{a}_2}$ do not have any overlap. 
\item  We then focus on the second condition of NHSDP, i.e.,  the half-sum of any two different elements in one block  does
not appear in any block of $\mathfrak{D}$.  For any two different integers in $\mathcal{D}_{{\bf a}}$, say $$d=\alpha_1a_1x_1+\alpha_2a_2x_2+\alpha_3a_3x_3\ \ \text{and}  \ \ d'=\alpha'_1a_1x_1+\alpha'_2a_2x_2+\alpha'_3a_3x_3,$$ where $(\alpha'_1,\alpha'_2,\alpha'_3)\in\{-1,1\}^3$ and $(\alpha_1,\alpha_2,\alpha_3)\neq (\alpha'_1,\alpha'_2,\alpha'_3)$, we have their half-sum 
$$\frac{d+d'}{2}=\frac{\alpha_1+\alpha'_1}{2}a_1x_1+\frac{\alpha_2+\alpha'_2}{2}a_2x_2+\frac{\alpha_3+\alpha'_3}{2}a_3x_3.$$ Recall that $(\alpha_1,\alpha_2,\alpha_3)\neq (\alpha'_1,\alpha'_2,\alpha'_3)$ there is at least one item of $\frac{\alpha_1+\alpha'_1}{2}$, $\frac{\alpha_2+\alpha'_2}{2}$ and $\frac{\alpha_3+\alpha'_3}{2}$ equal to $0$. This implies that each half-sum can only be represented by a linear combination of all elements in $\mathcal{X}$ with   coefficients in $\mathcal{C}'=[-2:2]^3\setminus\mathcal{C}$.   Hence, $\mathcal{C}\cap \mathcal{C}'=\emptyset$,  which implies that each half-sum does not occur in $\mathcal{D}$.
\end{itemize}

By Theorem~\ref{th-1} we have a $(125,125,61,1000)$ PDA which generates a $125$-division $(125,M,N)$ coded caching scheme with $M/N=\frac{61}{125}$, coded caching gain $g=8$, and the transmission load $b=8$. Now let us consider the existing schemes with $K=125$.
\begin{itemize}
\item  The WCWL scheme (see Table~\ref{tab-Schemes}): When $K=125$ and $t=61$ in \cite{WCWL} we have a WCWL scheme with $K=125$, the subpacketization $F=125$, and the transmission load $\frac{K-t}{2\lfloor \frac{K}{K-t+1}\rfloor}=\frac{64}{3}\approx21.333>8$. Clearly our scheme has a lower transmission load while maintaining a lower memory ratio and the same subpacketization;
\item  The CWZW scheme (see Table~\ref{tab-Schemes}): In the case of $q=2$ and $m=7$, the scheme  in \cite{CWZW} includes the schemes in \cite{YCTC,SZG,TR}. So we only need consider the CWZW scheme. First we can obtain a CWZW scheme with $K=16$, the memory ratio $M/N=0.5>\frac{61}{125}$, the subpacketization $F'=2^7=128>125$, and the transmission load $1$. Using the grouping method in Lemma~\ref{le-fundamental recursive}, we have a grouping CWZW scheme with $K=125$, $M/N=0.5>\frac{61}{125}$, the subpacketization $F_1=\frac{16}{\text{gcd}(125,16)}\cdot F'=2048>125$, and the transmission load $\frac{125}{16}=7.8125$. We can see that our scheme has  a lower  memory ratio and a smaller subpacketization than the scheme in \cite{CWZW}, albeit with a slight increase in the transmission load;
\item The MN scheme (see Table~\ref{tab-Schemes}): An exhaustive computer search shows that we can obtain a $(125,M,N)$ grouping MN coded caching scheme with $M/N=0.5$, an appropriate subpacketization ${10\choose 5}=252$, and the transmission load $R=\frac{5}{6}\times\frac{125}{10}=\frac{125}{12}>10>8$ by using a $(10,M,N)$ MN scheme based on the grouping method in Lemma~\ref{le-fundamental recursive}. It is evident that all the memory ratio, the subpacketization, and the transmission load are larger than our proposed scheme.
\end{itemize}

\hfill $\square$ 
\end{example} 
By Construction~\ref{constr-2}, for any  $n$ and $m_1$, $m_2$, $\ldots$, $m_n$ are positive integers, suppose the vector $\mathbf{a}=(a_1,a_2,\ldots,a_n)=(m_1,m_2,\ldots,m_n) \in \mathcal{A}$. The corresponding block  $\mathcal{D}_{\mathbf{a}}$ is given by  $$\mathcal{D}_{\mathbf{a}}=\{\alpha_1 f(1)+\alpha_2 f(2)+\cdots+\alpha_n f(n)\  |\  \alpha_i \in \{-1,1\}, i \in [n]\},$$
  which contains the integers   $-f(1)-f(2)-\dots-f(n)$ and $f(1)+f(2)+\ldots+f(n)$. To ensure that  $[-f(1)-f(2)-\dots-f(n):0]$ and $[1:f(1)+f(2)+\ldots+f(n)]$ can not overlap in $\mathbb{Z}_{v}$, the value of $v$ must satisfy $$v\geq 2(f(1)+f(2)+\ldots+f(n))+1.$$ 
For the ease of further notations, we define that
\begin{align}
\phi(m_1,m_2,\ldots,m_n):=&\sum_{i=1}^{n}f(i)=\sum_{i=1}^{n-1}f(i)+m_n\left(\sum_{j=1}^{n-1}2f(j)+1\right)
=(1+2m_n)\sum_{i=1}^{n-1}f(i)+m_n\nonumber\\
=&(1+2m_n)(1+2m_{n-1})\sum_{i=1}^{n-2}f(i)+m_{n-1}(1+2m_n)+m_n\nonumber\\
=&\prod_{i=n-2}^{n}(1+2m_i)\sum_{i=1}^{n-3}f(i)+m_{n-2}\prod_{i=n-1}^{n}(1+2m_i)+m_{n-1}(1+2m_n)+m_n\nonumber\\
=&\sum_{i=1}^{n-1}\left(m_i\prod_{j=i+1}^{n}(1+2m_{i})\right)+ m_{n}.
\label{eq-sum-half}
\end{align}
From \eqref{eq-sum-half} when $v\geq 2\phi(m_1,m_2,\ldots,m_n)+1$ we have the following lemma, whose proof is given in Appendix~\ref{sec:proof of thm2}.
\begin{lemma}\rm
\label{th-main-2}
For any $n$ and $m_1$, $m_2$, $\ldots$, $m_n$  are positive integers and for any odd positive integer $v\geq 2\phi(m_1,m_2,\ldots,m_n)+1$ defined in \eqref{eq-sum-half}, the pair $(\mathbb{Z}_{v},\mathfrak{D})$ generated in Construction~\ref{constr-2} is a $(v,2^n,\prod_{i=1}^{n}m_i)$ NHSDP.
\hfill $\square$ 
\end{lemma}

By using the above NHSDP construction, we can obtain the following PDA. 
\begin{theorem}\rm
\label{th-main-PDA}
For  any positive integer $n$ and $n$ positive integers $m_1$, $m_2$, $\ldots$, $m_n$,  given a  $(v,2^n,\prod_{i=1}^{n}m_i)$ NHSDP,  we can obtain a $v$-division $(v,M,N)$ coded caching scheme with memory ratio $M/N=1-\frac{2^n\prod_{i=1}^{n}m_i}{v}$, coded caching gain $g=2^n$, and transmission load $R=\prod_{i=1}^{n}m_i$, under the constraint of $v\geq 2\phi(m_1,m_2,\ldots,m_n)+1$. 
\hfill $\square$ 
\end{theorem}
Next, given the parameters $v$ and $n$ (recall that $K=v$ and $g=2^n$), we consider the selection of $m_1,\ldots,m_n$, such that the memory ratio is minimum. In other words, fixing the  coded caching gain, we aim to search the coded caching scheme requiring the minimum memory ratio.   
Considering the constraint $v\geq 2\phi(m_1,m_2,\ldots,m_n)+1$ into the optimization problem for the selection of $m_1,m_2,\ldots,m_n$, we have 
\begin{equation}\label{eq-optimization}
\begin{aligned}
	\text{Problem 1.}  \ \ \  & \textbf{Maximize } \text{function } f=\prod_{i=1}^{n}m_i \\
	&\textbf{Constrains: } m_1,\ldots,m_n \in \mathbb{Z}^{+},\\
	&\quad\quad\quad \quad \quad  \sum_{i=1}^{n-1}\left(m_i\prod_{j=i+1}^{n}(1+2m_{i})\right)+ m_{n} \leq \frac{v-1}{2}.
\end{aligned}
\end{equation}

Problem~1 is an integer programming problem, which is NP-hard. We first relax the constraint $m_1,m_2,\ldots,m_n \in \mathbb{Z}^{+}$ to $m_1,m_2,\ldots,m_n \in \mathbb{R}^{+}$, in order to simplify the problem to a convex optimization problem, which is then solved by the Lagrange Multiplier Method. Finally we take the floor operation to the solution.  For the ease of notation, we define that $q:=K^{1/n}$.  The following theorem introduces the sub-optimal solution by the above optimization strategy, whose proof is given in Appendix~\ref{sec:lagrange}.
\begin{theorem}\rm
\label{th:Lagrange}
A sub-optimal solution (with closed-form) to Problem~1 is 
\begin{align}
m_1=m_2=\cdots=m_n= \left\lfloor \frac{v^{1/n}-1}{2} \right\rfloor=\left\lfloor\frac{q-1}{2}\right\rfloor. 
\label{eq:m1=m2}
\end{align}
Under the selection in~\eqref{eq:m1=m2}, the resulting PDA has the parameters 
\begin{equation*}
K=q^n,\ \ F=q^n,\ \ Z=q^n-\left(2\left\lfloor \frac{q-1}{2} \right\rfloor\right)^n,\ \ S={\left\lfloor \frac{q-1}{2} \right\rfloor}^nq^n,
\end{equation*} which can realize a coded caching scheme with the memory ratio $\frac{M}{N}=1-\frac{2^n\lfloor\frac{q-1}{2}\rfloor^n}{q^n}$ and the transmission load $R={\left\lfloor \frac{q-1}{2} \right\rfloor}^n$.

\hfill $\square$ 
\end{theorem}
 By Theorem~\ref{th:Lagrange}, given a user number $K$ there are $\left \lfloor \log_p K\right \rfloor$ memory-rate ratio points for each positive $p\geq 3$. In addition, when $q$ is equal to a positive odd integer such that $q \geq 3$, the solution in \eqref{eq:m1=m2} is an optimal solution to Problem~1. Then by Theorem~\ref{th:Lagrange} we have the following result.
\begin{remark}[Optimal solution of Problem 1]\rm
\label{re-optimal}
When $q$ in Theorem~\ref{th:Lagrange} is an odd integer, we have a $(K=q^n, F=q^n,Z=q^n-(q-1)^n,S=(\frac{q-1}{2})^nq^n)$ PDA which realizes a $(K,M,N)$ coded cachign scheme with the memory ratio $M/N=1-(\frac{q-1}{q})^n$, subpacketization $F=q^n$, and transmission load $R=(\frac{q-1}{2})^n$.
\end{remark}

Given a PDA, the authors in \cite{CJTY} pointed out that its corresponding conjugate PDAs can be obtained in the following.
\begin{lemma}\cite{CJTY}\rm 
\label{permutations of PDA}
Gvien a $(K,F,Z,S)$ PDA for some positive integers $K$, $F$, $Z$ and $S$ with $0< Z<F$, there exists a $(K,S,S-(F-Z),F)$ PDA. \hfill $\square$  
\end{lemma}
By Theorem~\ref{th:Lagrange} and Lemma~\ref{permutations of PDA}, the following new PDA can be obtained. 

\begin{corollary}[Conjugate PDA of Theorem~\ref{th:Lagrange}]\rm
\label{cor:conjugate}
The conjugate PDA corresponding  to the PDA mentioned in Theorem~\ref{th:Lagrange} is a  $(q^n,{\left\lfloor \frac{q-1}{2} \right\rfloor}^nq^n,{\left\lfloor \frac{q-1}{2} \right\rfloor}^n(q^n-2^n),q^n)$ PDA which can realize a coded caching scheme with memory ratio $\frac{M}{N}=1-(\frac{2}{q})^n$ and the transmission load $R={\left\lfloor \frac{q-1}{2} \right\rfloor}^{-n}$.
\end{corollary}
In particular, if $q$ equals to some odd integer such that $q \geq 3$, we have PDAs in  above Corollary~\ref{cor:conjugate} can be written as the following $(q^n,(\frac{q-1}{2})^nq^n,(\frac{q-1}{2})^nq^n-(q-1)^n,q^n)$ PDA.

\section{Performance analysis}
\label{sec-perf-ana}
In this section, we will present theoretical and numerical comparisons with the existing schemes respectively to show the performance of our new scheme in Theorem~\ref{th:Lagrange}.

\subsection{Theoretical comparisons} 
Since the schemes in \cite{YTCC,CKSM,ASK,ZCW,AST} have the special parameters of the user number and the memory ratio we only need to compare our scheme with the schemes in~\cite{MN,YCTC,SZG,CWZW,WCLC,WCWL,XXGL} respectively.
\subsubsection{Comparison with the MN scheme in \cite{MN}}  
When $t=q^n-2^n\lfloor\frac{q-1}{2}\rfloor^n$,  we can obtain the MN scheme with the memory ratio $\frac{M}{N}=1-\frac{2^n\lfloor\frac{q-1}{2}\rfloor^n}{q^n}$, the subpacketization $F_{\text{MN}}={K\choose t}=\binom{q^n}{q^n-2^n\lfloor\frac{q-1}{2}\rfloor^n}=\binom{q^n}{2^n\lfloor\frac{q-1}{2}\rfloor^n}$, and the transmission load $R_{\text{MN}}=\frac{2^n\lfloor\frac{q-1}{2}\rfloor^n}{q^n-2^n\lfloor\frac{q-1}{2}\rfloor^n+1}$. By Theorem~\ref{th:Lagrange}, we can obtain our scheme with the same memory ratio where  $K=F=q^n$ and the  transmission load $R=\lfloor\frac{q-1}{2}\rfloor^n$. Then  the following results can be obtained,
\begin{align*}
\frac{F_{\text{MN}}}{F}
&=\frac{\binom{q^n}{2^n\lfloor\frac{q-1}{2}\rfloor^n}}{q^n}
\approx\frac{q^{n2^n\lfloor\frac{q-1}{2}\rfloor^n}}{q^n} >\frac{q^{n(q-3)^n}}{q^n}=q^{n(q-3)^n-n}=K^{(q-3)^n-1},\\
\frac{R_{\text{MN}}}{R}
&=\frac{\frac{2^n\lfloor\frac{q-1}{2}\rfloor^n}{q^n-2^n\lfloor\frac{q-1}{2}\rfloor^n+1}}{\lfloor\frac{q-1}{2}\rfloor^n}
=\frac{2^n}{q^n-2^n\lfloor\frac{q-1}{2}\rfloor^n+1}
<\frac{2^n}{q^n-(q-1)^n+1}
\approx \frac{1}{K}\cdot\frac{2^n}{1-(\frac{q-1}{q})^n}.
\end{align*}Compared to the MN scheme, the multiplicative reduction amount by our subpacketization is proportional to  $K^{(q-3)^n-1}$ which grows exponentially with the user number $K$ at a rate of  $(q-3)^n$, while the increase amount of our transmission is only  $$\frac{q^n-2^n\lfloor\frac{q-1}{2}\rfloor^n+1}{2^n}\approx \frac{q^n-(q-1)^n}{2^n}.$$

\subsubsection{Comparison with the WCLC scheme in \cite{WCLC}} 
The authors in \cite{WCLC} showed that the scheme includes the partition schemes in \cite{YCTC}, the hypergraph schemes in \cite{SZG}, and the OA schemes in \cite{CWZW} as the special cases.  So first we only need to compare with  the scheme in \cite{WCLC}. When $z>1$,  the user number in the WCLC scheme is too complex, so we have to use some specific parameters for the comparison. When $m=z=n$, $k=q$, and $t=q-2\lfloor\frac{q-1}{2}\rfloor$, we have the WCLC scheme in  \cite{WCLC}, where the number of users $K=q^n$ with the memory ratio, the subpacketization, and the transmission load are respectively
\begin{align*}
\frac{M}{N}=1-\frac{2^n\lfloor\frac{q-1}{2}\rfloor^n}{q^n},\ \ \ F_{\text{WCLC}}=\frac{q^{n-1}(q-1)^n}{2^n\lfloor\frac{q-1}{2}\rfloor^n},\ \ \  R_{\text{WCLC}}=\frac{2^n\left\lfloor\frac{q-1}{2}\right\rfloor^n}{\lfloor\frac{q-1}{2\left\lfloor\frac{q-1}{2}\right\rfloor}\rfloor^n}.
\end{align*}By Theorem~\ref{th:Lagrange}, we can obtain a coded caching scheme with the same memory ratio and the same number of users, the subpacketization $F=q^n$, and the transmission load $R=\lfloor\frac{q-1}{2}\rfloor^n$. Then the following result can be obtained,
\begin{align*}
\frac{F_{\text{WCLC}}}{F}=\frac{\frac{q^{n-1}(q-1)^n}{2^n\lfloor\frac{q-1}{2}\rfloor^n}}{q^n}=\frac{(q-1)^n}{q\cdot2^n\lfloor\frac{q-1}{2}\rfloor^n} \approx\frac{1}{q},\ \ \ 
\frac{R_{\text{WCLC}}}{R}=\frac{\frac{2^n\lfloor\frac{q-1}{2}\rfloor^n}{\lfloor\frac{q-1}{2\lfloor\frac{q-1}{2}\rfloor}\rfloor^n}}{\left\lfloor\frac{q-1}{2}\right\rfloor^n}=\frac{2^n}{\lfloor\frac{q-1}{2\lfloor\frac{q-1}{2}\rfloor}\rfloor^n} < \frac{4^{n}\lfloor\frac{q-1}{2}\rfloor^n}{(q-1)^n}\approx2^n.
\end{align*}Compared to the WCLC scheme, our subpacketization increases by a factor of $q$  while reducing  the transmission load by a factor of $2^n$.

\subsubsection{Comparison with the WCWL scheme in \cite{WCWL}} 
 Let $K=q^n$ and $t=q^n-2^n \lfloor\frac{q-1}{2}\rfloor^n$ in \cite{WCWL}, i.e., the $10$th row of Table~\ref{tab-Schemes}. 
We compare the WCWL scheme with the proposed scheme in the following cases:
\begin{itemize}
\item If $(K-t+1)|K$ or $K-t=1$, we have 
\begin{equation*}
F_{\text{WCWL}}=K=q^n,\ \  R_{\text{WCWL}}=\frac{(K-t)(K-t+1)}{2K}=\frac{2^n\lfloor\frac{q-1}{2}\rfloor^n(2^n\lfloor\frac{q-1}{2}\rfloor^n+1)}{2q^n}.
\end{equation*}
 By Theorem~\ref{th:Lagrange}, we have the following result.
\begin{align*}
\frac{F_{\text{WCWL}}}{F}&=\frac{q^n}{q^n}=1, \\
\frac{R_{\text{WCWL}}}{R}&=\frac{\frac{2^n\lfloor\frac{q-1}{2}\rfloor^n(2^n\lfloor\frac{q-1}{2}\rfloor^n+1)}{2q^n}}{\lfloor\frac{q-1}{2}\rfloor^n}=\frac{2^n(2^n\lfloor\frac{q-1}{2}\rfloor^n+1)}{2q^n}>\frac{2^n(q-3)^n}{2q^n}=\frac{1}{2}\left(2-\frac{6}{q}\right)^n.
\end{align*}
\item If $\left \langle K \right \rangle_{K-t+1}=K-t$, we can get the WCWL scheme with
 \begin{equation*}
 	F_{\text{WCWL}}=\left(2\lfloor\frac{K}{K-t+1}\rfloor+1\right)K,\ \ \ R_{\text{WCWL}}=\frac{K-t}{2\lfloor\frac{K}{K-t+1}\rfloor+1}
 	=\frac{2^n\lfloor\frac{q-1}{2}\rfloor^n}{2\left\lfloor \frac{q^n}{2^n\lfloor\frac{q-1}{2}\rfloor^n+1} \right\rfloor+1}.
 \end{equation*}
 By Theorem~\ref{th:Lagrange}, we can obtain the following.
 \begin{align*}
 \frac{F_{\text{WCWL}}}{F}&=\frac{(2\lfloor \frac{K}{K-t+1}\rfloor+1) K}{q^n}=2\left\lfloor \frac{K}{K-t+1}\right\rfloor+1,\\
 \frac{R_{\text{WCWL}}}{R}&=\frac{\frac{2^n\lfloor\frac{q-1}{2}\rfloor^n}{2\left\lfloor \frac{q^n}{2^n\lfloor\frac{q-1}{2}\rfloor^n+1} \right\rfloor+1}}{\lfloor\frac{q-1}{2}\rfloor^n}
 =\frac{2^n}{2\left\lfloor \frac{q^n}{2^n\lfloor\frac{q-1}{2}\rfloor^n+1} \right\rfloor+1}\\
 &>\frac{2^n(q-3)^n}{2q^n+(q-1)^n+1}\\
 &>\frac{2^n(q-3)^n}{3q^n}\\
& =\frac{1}{3}\cdot\left(2-\frac{6}{q}\right)^n.
 \end{align*}
 \item Otherwise, the WCWL scheme  can be obtain as follow.
   $$F_{\text{WCWL}}=2\left\lfloor\frac{K}{K-t+1}\right\rfloor K,\ \ \ R_{\text{WCWL}}=\frac{K-t}{2\lfloor\frac{K}{K-t+1}\rfloor}
  =\frac{2^n\lfloor\frac{q-1}{2}\rfloor^n}{2\left\lfloor \frac{q^n}{2^n\lfloor\frac{q-1}{2}\rfloor^n+1} \right\rfloor}.$$ By Theorem~\ref{th:Lagrange}, we have
  \begin{align*}
  \frac{F_{\text{WCWL}}}{F}&=\frac{2\left\lfloor\frac{K}{K-t+1}\right\rfloor K}{q^n}=2\left\lfloor \frac{K}{K-t+1}\right\rfloor,\\
  \frac{R_{\text{WCWL}}}{R}&=\frac{\frac{2^n\lfloor\frac{q-1}{2}\rfloor^n}{2\left\lfloor \frac{q^n}{2^n\lfloor\frac{q-1}{2}\rfloor^n+1} \right\rfloor}}{\lfloor\frac{q-1}{2}\rfloor^n}
  =\frac{2^n}{2\left\lfloor \frac{q^n}{2^n\lfloor\frac{q-1}{2}\rfloor^n+1} \right\rfloor}
  >\frac{2^n(q-3)^n}{2q^n}=\frac{1}{2}\cdot \left(2-\frac{6}{q}\right)^n.
  \end{align*}   
\end{itemize}

As a result, compared to the WCWL scheme in \cite{WCWL},  our scheme either obtaining the  same subpacketization or the reduction amount of our subpacketization is larger than $2$ times.  The multiplicative reduction amount by our transmission is at least $\frac{1}{3}\cdot(2-\frac{6}{q})^n$. When $q>6$ and $n \geq \frac{\ln3}{\ln (2-\frac{6}{q})}$, the formula $\frac{1}{3}\cdot(2-\frac{6}{q})^n\geq 1$ always holds.

\subsubsection{Comparison with the XXGL scheme in \cite{XXGL}} 
Let $K=q^n$ in \cite{XXGL}, i.e., the $11$th row of Table~\ref{tab-Schemes}. We have the XXGL scheme with memory ratio $M/N=1-\frac{2}{q^n}\gg 1-(\frac{2}{q})^n$, the subpacketization 
$F_{\text{XXGL}}=q^n$, and the transmission load $R_{\text{XXGL}}=\frac{q^n-1}{q^n}$. The scheme in  Corollary~\ref{cor:conjugate} have $F=\lfloor\frac{q-1}{2}\rfloor^{n}q^n$ and $R=\lfloor\frac{q-1}{2}\rfloor^{-n}$. So we have 
\begin{align*}
\frac{F_{\text{XXGL}}}{F}=\frac{q^n}{\lfloor\frac{q-1}{2}\rfloor^{n}q^n}=\left\lfloor\frac{q-1}{2}\right\rfloor^{-n},\ \ \ \ 
\frac{R_{\text{XXGL}}}{R}=\frac{\frac{q^n-1}{q^n}}{\lfloor\frac{q-1}{2}\rfloor^{-n}}\approx\left\lfloor\frac{q-1}{2}\right\rfloor^n.
\end{align*}
Compared to the XXGL scheme in \cite{XXGL}, even though our memory ratio is much smaller than that of the XXGL scheme, the reduction amount of our transmission load is about $\lfloor\frac{q-1}{2}\rfloor^n$ times while our subpacketization increases  $\lfloor\frac{q-1}{2}\rfloor^{n}$ times which is much less than $K$. So our scheme has the significant advantages on transmission load compared to the XXGL scheme.  

\subsection{Numerical comparisons }\label{subsec-p}
In this subsection, we will present numerical comparisons between our scheme  in Theorem~\ref{th:Lagrange} and the existing schemes with linear subpacketization \cite{XXGL,WCWL,AST,ZCW,ASK}, polynomial subpacketization \cite{CKSM,YTCC}, and exponential subpacketization \cite{MN,WCLC}. These existing schemes are summarized in Table~\ref{tab-Schemes}. It is worth noting that the ASK scheme \cite{ASK} can be included as a special case of our scheme in Corollary~\ref{coro-first-q} which will be introduced in Subsection~\ref{subsect-CDP}. In addition, the XXGL scheme \cite{XXGL} has large memory ratio which approximates $1$ so in this subsection we do not need to consider the XXGL scheme.

\subsubsection{Comparison with the linear subpacketization schemes in \cite{AST,ZCW,WCWL}}
Since these schemes have the special integers, including combination numbers, powers, products of combination numbers, and powers, it is difficult for us to directly plot a graph for given values of $K$. So we present some numerical comparisons between our scheme in Theorem \ref{th:Lagrange} and the schemes in \cite{AST,ZCW} in  Table~\ref{tab-numerical-2}. By Table ~\ref{tab-numerical-2}, our scheme in Theorem~\ref{th:Lagrange} has more users, a close memory ratio, a similar subpacketization, and a smaller load than that of the  ZCW scheme; our scheme has a smaller memory ratio, a lower transmission load, and a smaller subpacketization than that of the AST scheme,  while having more users.   
\begin{table}[htbp!]
	\centering
	\caption{The numerical comparison between  the scheme in Theorem~\ref{th:Lagrange} and the schemes in \cite{AST,ZCW}}
	\label{tab-numerical-2}
	\begin{tabular}{|c|c|c|c|c|c|}
		\hline
		$K$   & $M/N$   & Scheme  & Parameters & Load   & Subpacketization \\ \hline 
		$32$ & $0.6875$ & ZCW scheme in \cite{ZCW} & $(m,w)=(5,2)$ & $1.25$ & $32$  \\
		$33$ & $0.7576$ & Scheme in Theorem~\ref{th:Lagrange}& $(v,n)=(33,3)$ & $1$ & $33$  \\ \hline
		$128$ & $0.836$ & ZCW scheme in \cite{ZCW} & $(m,w)=(7,2)$ & $2.625$ & $128$  \\ 
		$129$ & $0.876$ & Scheme in Theorem~\ref{th:Lagrange} & $(v,n)=(129,4)$ & $1$ & $129$  \\ \hline
		$256$ & $0.8906$ & ZCW scheme in \cite{ZCW} & $(m,w)=(8,2)$ & $2.41$ & $256$  \\ 
		$257$ & $0.8755$ & Scheme in Theorem~\ref{th:Lagrange} & $(v,n)=(257,5)$ & $1$ & $257$  \\ \hline
		$512$ & $0.9297$ & ZCW scheme in \cite{ZCW} & $(m,w)=(9,2)$ & $2.04$ & $512$  \\ 
		$513$ & $0.9376$ & Scheme in Theorem~\ref{th:Lagrange} & $(v,n)=(513,5)$ & $1$ & $513$  \\ \hline
		$52$ & $0.28846$ & AST scheme in \cite{AST} & $(r,k)=(2,13)$ & $9.25$ & $52$  \\ 
		$49$ & $0.26531$ & Scheme in Theorem~\ref{th:Lagrange}  & $(v,n)=(49,2)$ & $9$ & $49$  \\ \hline 
		$1332$ & $0.2515$ & AST scheme in \cite{AST}  & $(r,k)=(2,333)$ & $249.25$ & $1332$  \\
		$1331$ & $0.2498$ & Scheme in Theorem~\ref{th:Lagrange}  & $(v,n)=(1331,3)$ & $125$ & $1331$  \\  \hline
		
		$2192$ & $0.2509$ & AST scheme in \cite{AST}  & $(r,k)=(2,548)$ & $410.5$ & $2192$  \\
		$2199$ & $0.2142$ & Scheme in Theorem~\ref{th:Lagrange} & $(v,n)=(2199,3)$ & $216$ & $2199$  \\ \hline
		
		$2400$ & $0.50125$ & AST scheme in \cite{AST} & $(r,k)=(3,300)$ & $149.625$ & $2400$  \\
		$2401$ & $0.460$ & Scheme in Theorem~\ref{th:Lagrange}  & $(v,n)=(2401,4)$ & $81$ & $2401$  \\  \hline
	\end{tabular}
\end{table}

\begin{figure}[htbp!]
	\centering
	\begin{minipage}[b]{.5\textwidth}
		\centering
		\includegraphics[scale=0.6]{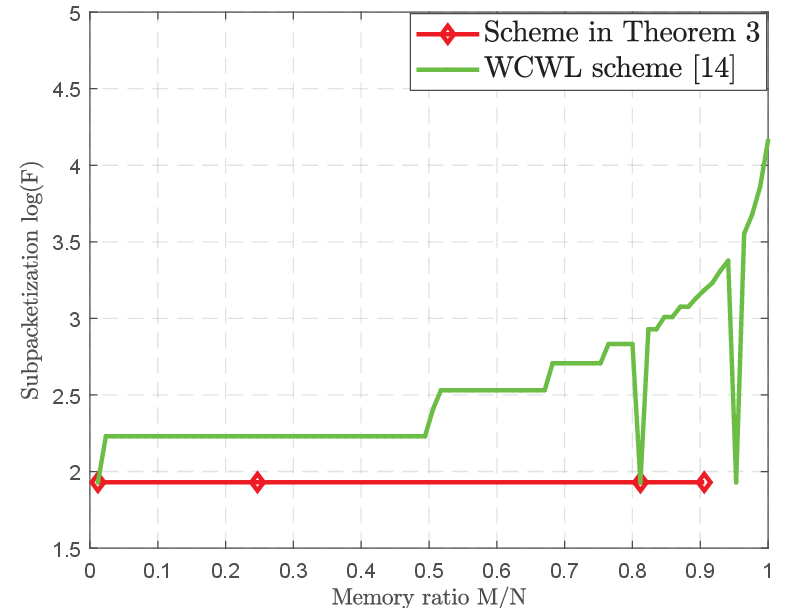}
		\caption{Memory ratio-subpacketization tradeoff for $K=85$} 
		\label{fig-com9}
	\end{minipage}
	\begin{minipage}[b]{.45\textwidth}
		\centering
		\includegraphics[scale=0.6]{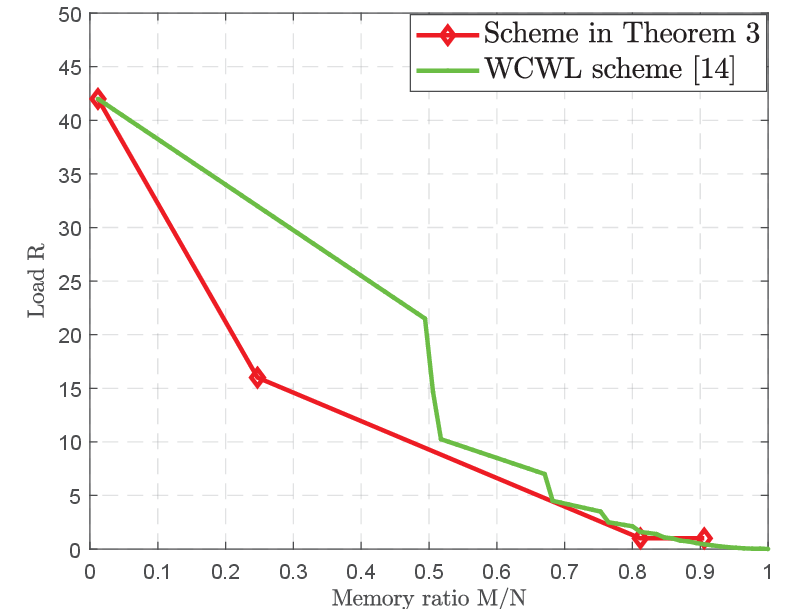}
		\caption{Memory ratio-load tradeoff for $K=85$} 
		\label{fig-com10}
	\end{minipage}
\end{figure}
Now, let us compare our scheme in Theorem~\ref{th:Lagrange} with the WCWL scheme in \cite{WCWL}. When $K=85$, we have our scheme (the red line) and the WCWL scheme (the  green line) listed in  Figure~\ref{fig-com9} and Figure~\ref{fig-com10}. We can see that our scheme has a smaller or equal subpacketization, and a lower transmission load than that of the WCWL scheme when the memory ratio is less than or equal to $0.8$. When the memory ratio is close to $0.9$, our scheme has a slightly higher transmission load while having a smaller or equal subpacketization.  
	
Finally, we would like to point out that the advantages of our scheme in terms of the subpacketization and transmission load become more obvious as the number of users $K$ increases. For instance, when $K=729$, we have our scheme and the WCWL scheme listed in Figure~\ref{fig-com11} and Figure~\ref{fig-com12}. We observe that our scheme has a smaller or equal subpacketization and a much smaller transmission load. 
\begin{figure}[htbp!]
	\centering
	\begin{minipage}[b]{.5\textwidth}
		\centering
		\includegraphics[scale=0.6]{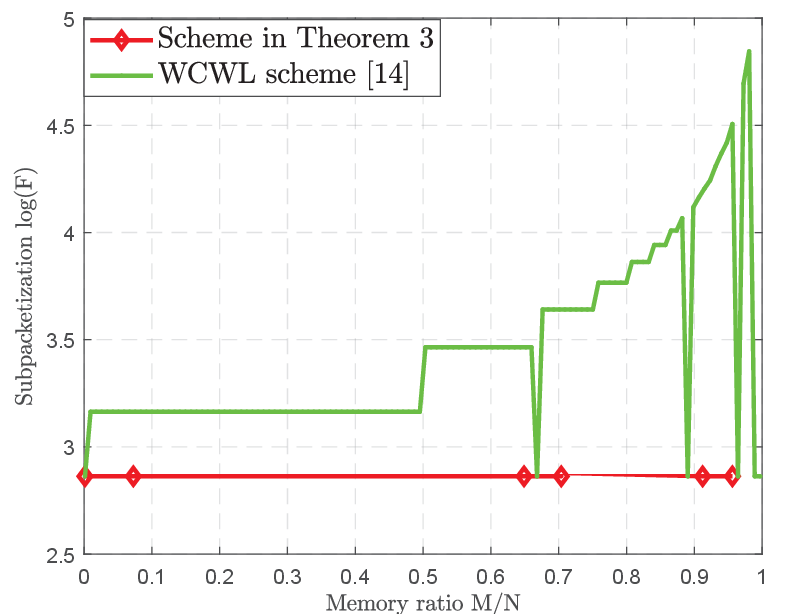}
		\caption{Memory ratio-subpacketization tradeoff for $K=729$} 
		\label{fig-com11}
		
	\end{minipage}
	\begin{minipage}[b]{.4\textwidth}
		\centering
		\includegraphics[scale=0.6]{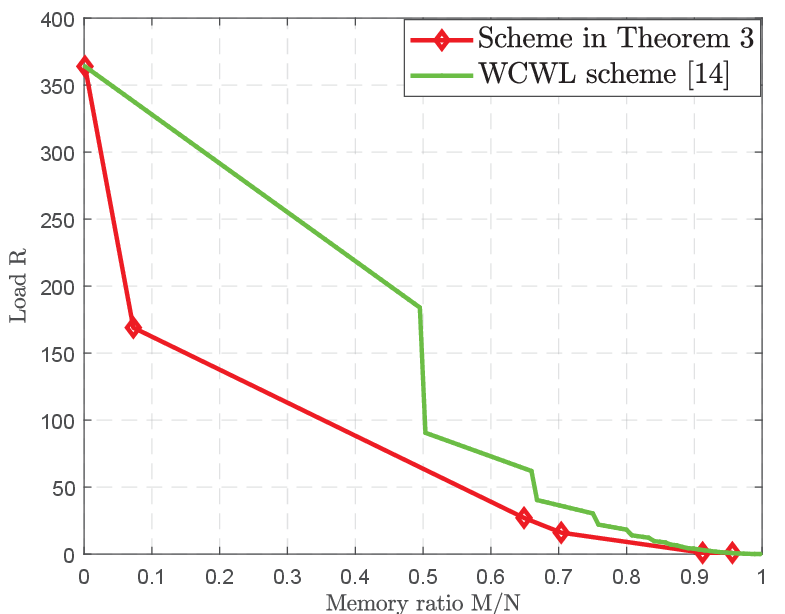}
		\caption{Memory ratio-load tradeoff for $K=729$} 
		\label{fig-com12}
	\end{minipage}
\end{figure}

\subsubsection{Comparison with the polynomial subpacketization schemes in \cite{CKSM,YTCC}}
Since the user numbers of the schemes in \cite{CKSM,YTCC} are the special integers, including combination numbers, powers, products of combination numbers, and powers, it is difficult for us to directly plot a graph for given values of $K$. Instead, we compare our scheme with the CKSM scheme and the YTCC scheme by choosing some specific user numbers and memory ratios listed in Table~\ref{tab-numerical-1}. We can see that compared to the YTCC scheme, our scheme has a similar memory ratio, a slightly larger transmission load, and more users, while having a much smaller subpacketization; compared to the CKSM scheme, our scheme has a lower transmission load and a much smaller subpacketization while having a close number of users and a close memory ratio.

\begin{table}[htbp!]
\centering
\caption{The numerical comparison between the schemes in \cite{CKSM,YTCC} and the scheme in Theorem~\ref{th:Lagrange}}
\label{tab-numerical-1}
\begin{tabular}{|c|c|c|c|c|c|}
\hline
$K$   & $M/N$   & Scheme  & Parameters & Load   & Subpacketization \\ \hline

$22$   & $0.364$  & YTCC scheme in \cite{YTCC}   & $(H,a,b,r)=(22,1,8,0)$  & $1.556$ & $319770$           \\
$25$   & $0.36$   & Scheme in Theorem~\ref{th:Lagrange}  & $(v,n)=(25,2)$          & $4$ & $25$                   \\ \hline

$78$   & $0.81$   & YTCC scheme in \cite{YTCC}   & $(H,a,b,r)=(13,2,7,0)$     & $0.41667$ & $1716$       \\ 
$81$   & $0.802$  & Scheme in Theorem~\ref{th:Lagrange}  & $(v,n)=(81,4)$               & $1$ & $81$              \\ \hline

$105$  & $0.486$  & YTCC scheme in \cite{YTCC}   & $(H,a,b,r)=(15,2,9,1)$   & $6$    & $5005$            \\
$127$ & $0.496$ & CKSM  scheme in \cite{CKSM}& $(q,k,m,t)=(2,7,6,1)$ & $9.143$ & $3.56E+09$   \\ 
$125$  & $0.488$  & Scheme in Theorem~\ref{th:Lagrange}  & $(v,n)=(125,3)$             & $8$    & $125$            \\ \hline

$325$  & $0.354$  & YTCC scheme in \cite{YTCC}   & $(H,a,b,r)=(26,2,5,0)$ & $10$  & $65780$              \\
$343$  & $0.370$  & Scheme in Theorem~\ref{th:Lagrange} & $(v,n)=(343,3)$               & $27$ & $343$          \\\hline

$231$  & $0.805$  & YTCC scheme in \cite{YTCC}   & $(H,a,b,r)=(22,2,12,0)$ & $0.4945055$  & $646646$     \\
$255$ & $0.8784$ & CKSM scheme in   \cite{CKSM} & $(q,k,m,t)=(2,8,4,1)$ & $2.0667$ & $97155$  \\ 
$243$ & $0.8683$ & Scheme in Theorem~\ref{th:Lagrange} & $(v,n)=(243,5)$ & $1$ & $243$  \\  
\hline  
$31$ & $0.7742$ & CKSM scheme in    \cite{CKSM}& $(q,k,m,t)=(2,5,2,1)$ & $1$ & $155$  \\ 
$27$ & $0.7037$ & Scheme in Theorem~\ref{th:Lagrange} & $(v,n)=(27,3)$ & $1$ & $27$  \\ \hline

$85$ & $0.247$ & CKSM  scheme in \cite{CKSM}& $(q,k,m,t)=(4,4,3,1)$ & $16$& $95200$   \\ 
$81$ & $0.210$ & Scheme in Theorem~\ref{th:Lagrange} & $(v,n)=(81,2)$ & $16$  & $81$  \\ \hline

$156$ & $0.19872$ & CKSM  scheme in \cite{CKSM}& $(q,k,m,t)=(5,4,3,1)$ & $31.25$ & $604500$   \\ 
$121$ & $0.1736$ & Scheme in Theorem~\ref{th:Lagrange} & $(v,n)=(121,2)$ & $25$ & $121$   \\ \hline

$156$ & $0.1987$ & CKSM scheme in   \cite{CKSM} & $(q,k,m,t)=(5,4,3,1)$ & $31.25$ & $604500$  \\ 
$169$ & $0.1479$ & Scheme in Theorem~\ref{th:Lagrange} & $(v,n)=(169,2)$ & $36$ & $169$  \\ \hline

$255$ & $0.12157$ & CKSM  scheme in \cite{CKSM}& $(q,k,m,t)=(2,8,5,1)$ & $37.333$ & $8.10E+09$   \\ 
$225$ & $0.1289$ & Scheme in Theorem~\ref{th:Lagrange} & $(v,n)=(225,2)$ & $49$ & $225$   \\ \hline

$364$ & $0.332$ & CKSM scheme in \cite{CKSM}& $(q,k,m,t)=(3,6,5,1)$  & $40.5$   & $4.51E+10$ \\ $343$ & $0.370$ & Scheme in Theorem~\ref{th:Lagrange} & $(v,n)=(343,3)$ & $27$  & $343$   \\ \hline
\end{tabular}
\end{table}

\subsubsection {Comparison with the  exponential subpacketization schemes in \cite{MN,WCLC}} 
Let us first directly compare our scheme in Theorem~\ref{th:Lagrange} with the MN scheme and the WCLC scheme, each of which has an exponential subpacketization level with the number of users. When $K=85$, we have our scheme (the red line), the MN scheme (the blue line), and the WCLC scheme (the green line) listed in Figure~\ref{fig-com1} and Figure~\ref{fig-com2} by Table~\ref{tab-Schemes} and Theorem~\ref{th:Lagrange}. 

Compared to the MN scheme, we can observe that our  scheme (the red line)  has a much smaller subpacketization  than the MN scheme (the blue line),  while our transmission load is little larger than the MN scheme when $K=85$. Compared to the WCLC scheme, our scheme (the red line) has a much smaller subpacketization than the WCLC scheme (the green line) while our transmission load is close to or equal to that of the WCLC scheme when $K=85$.
\begin{figure}[htbp!]
	\centering
	\begin{minipage}[t]{0.5\textwidth}
		\centering
		\includegraphics[scale=0.6]{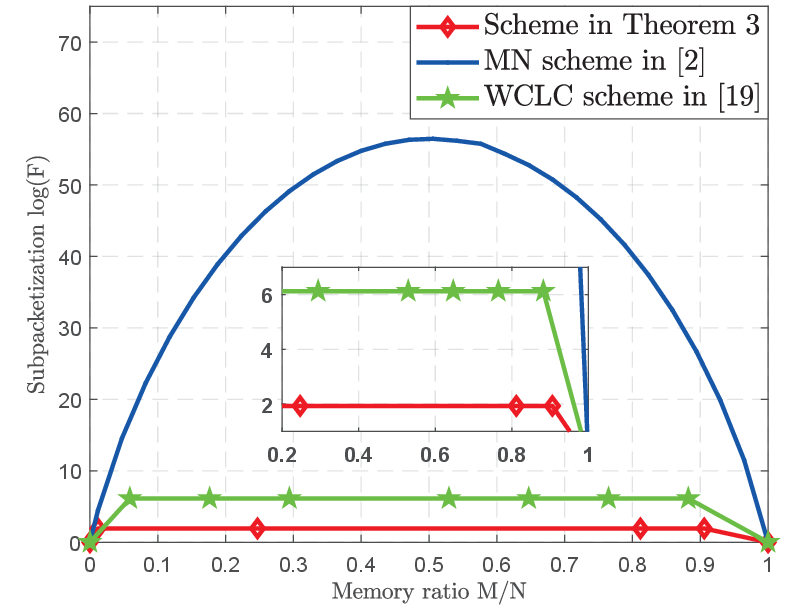}
		\caption{Memory ratio-subpacketization tradeoff for $K=85$} 
		\label{fig-com1}
	\end{minipage}
	\begin{minipage}[t]{0.45\textwidth}
		\centering
		\includegraphics[scale=0.6]{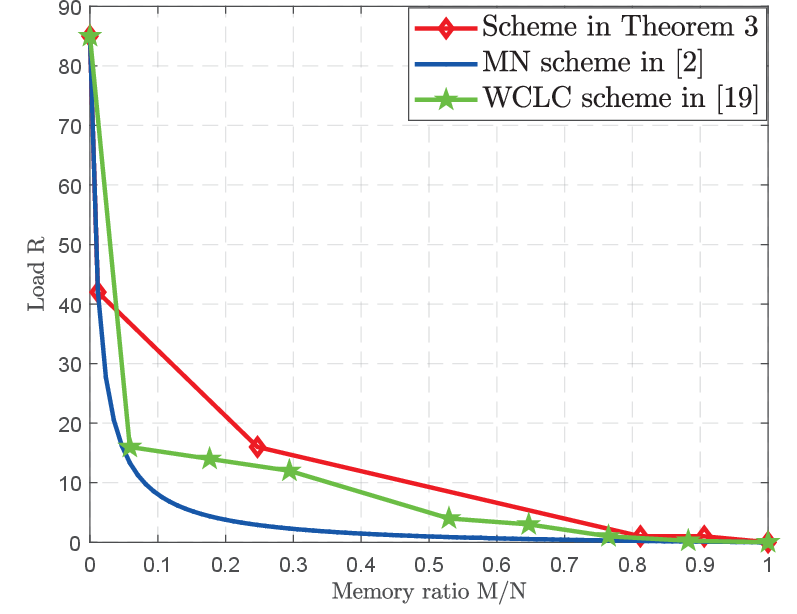}
		\caption{Memory ratio-load tradeoff for $K=85$} 
		\label{fig-com2}
	\end{minipage}
\end{figure}

Using the grouping method in Lemma~\ref{le-fundamental recursive} on the schemes in \cite{MN,WCLC},  through an exhaustive search on $K_1$ given $K=85$, we find that we cannot obtain a scheme with both a lower load and  a smaller subpacketization. 
Finally we should point out that the advantages on the subpacketization and transmission of our scheme are more obvious as $K$ increases. For instance, when $K=729$ we compare our scheme with the MN scheme and the WCLC scheme in Fig.~\ref{fig-com3} and Fig.~\ref{fig-com4}. Note that the ratio of the subpacketization between our scheme and the WCLC scheme significantly reduces (see Fig.~\ref{fig-com1} and Fig.~\ref{fig-com3}), and our transmission load is more close to that of the WCLC  (see Fig.~\ref{fig-com2} and Fig.~\ref{fig-com4}). 
\begin{figure}[htbp!]
	\centering
	
	\begin{minipage}[b]{.5\textwidth}
		\centering
		\includegraphics[scale=0.6]{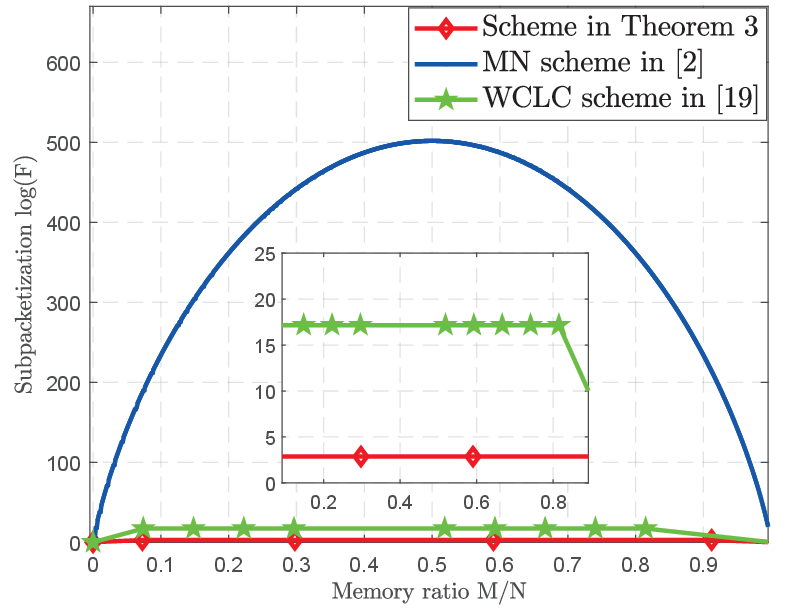}
		\caption{Memory ratio-subpacketization tradeoff for $K=729$ } 
		\label{fig-com3}
	\end{minipage}
	\begin{minipage}[b]{.45\textwidth}
		\centering
		\includegraphics[scale=0.6]{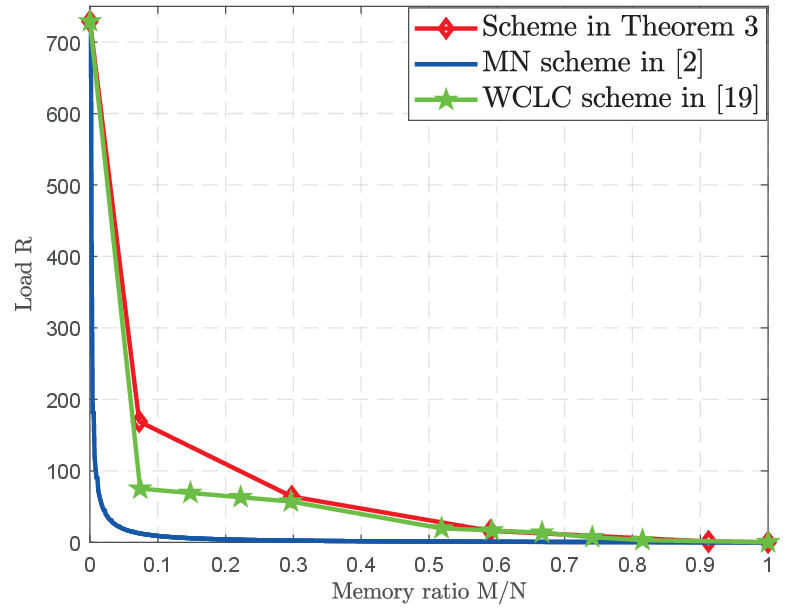}
		\caption{Memory ratio-load tradeoff  for $K=729$ } 
		\label{fig-com4}
	\end{minipage}
	
\end{figure}

\section{Relationship between NHSDP and the classic combinatorial Designs} 
\label{sec-expanding}
In this section, we will show that there is a close relationship between NHSDP and several classical combinatorial designs, such as cyclic difference packing \cite{dm/Yin98}, non-three-term arithmetic progressions \cite{brown1982density}, and perfect hash families \cite{stinson2006cryptography}. As a result, some new NHSDPs can be obtained by the existing cyclic difference packings,  and some new classes of non-three-term arithmetic progression sets and some new perfect hash families are obtained by the NHSDPs in Theorem~\ref{th:Lagrange}.
\subsection{The cyclic difference packing}
\label{subsect-CDP}
\begin{definition}[CDP\cite{dm/Yin98}]\rm
\label{def-CDP}
A $(v,k)$ cyclic difference packing (CDP) is a pair $(\mathbb{Z}_{v},\mathcal{D})$ where the subset $\mathcal{D}=\{d_1,d_2,\ldots,d_k\}$ of  $\mathbb{Z}_{v}$ satisfies that each non-zero integer has at most one representation as a difference $d_i-d_j$. If  each non-zero integer has exactly one representation as a difference $d_i-d_j$,  $(\mathbb{Z}_{v},\mathcal{D})$ is called $(v,k)$ difference set (DS).\hfill $\square$
\end{definition}
For instance when $v=7$ and $k=3$, we can obtain a pair $(\mathbb{Z}_{7},\mathcal{D}=\{0,1,3\})$. In addition we have the following relationships.
$$1=1-0,\ 2=3-1,\ 3=3-0,\ 4=0-3,\ 5=0-2,\ 6=0-1,$$
i.e., each  non-zero integer has exactly one representation as a difference $d_i-d_j$ for each $1\leq i\neq j\leq 3$. So  $(\mathbb{Z}_{7},\mathcal{D})$ is a $(7,3)$ CDP and is also a $(7,3)$ DS. In addition we can check that $(\mathbb{Z}_{7},\mathcal{D})$  is also a $(7,3,1)$ NHSDP since $\{\frac{0+1}{2}=4, \frac{0+3}{2}=5,\frac{1+3}{2}=2\}\cap\{0,1,3\}=\emptyset$. In general, the following statement always holds.
\begin{lemma}\rm
\label{PCDF-NHSDP}
Any $(v,k)$ CDP is a $(v,g=k,1)$ NHSDP. \hfill $\square$
\end{lemma} 
\begin{proof}
Assume that $(\mathbb{Z}_{v},\mathcal{D}=\{d_1,d_2,\ldots,d_k\})$ is a $(v,k)$ CDP. In order to show that it is also an NHSDP, we only need to consider the half-sum of any two different elements in $\mathcal{D}$. Suppose that there are three different elements, say $x$, $y$ and $z$, of $\mathcal{D}$ that satisfy  $2z=x+y$. Then we have $a=z-y=x-z$ which is impossible by Definition~\ref{def-CDP}, i.e., each non-zero integer has at most one representation as a difference of two different elements from  $\mathcal{D}$. Then our proof is complete.
\end{proof} 
It is well known that for any prime power $q$, there always exists a $(q^2+q+1,q+1)$ DS. By Lemma~\ref{PCDF-NHSDP} we have a $(q^2+q+1,q+1,1)$ NHSDP. Then by Theorem~\ref{th-1} and Lemma~\ref{le-PDA} the following result can be obtained.
\begin{corollary}\rm
\label{coro-first-q}
For any prime power $q$, there exists a $(q^2+q+1,q^2+q+1,q+1,q^2+q+1)$ PDA which realizes a $(q^2+q+1)$-division $(K=q^2+q+1,M,N)$ coded caching scheme with $M/N=\frac{q^2}{q^2+q+1}$ and transmission load $R=1$.\hfill $\square$
\end{corollary}

It is worth noting that the scheme in Corollary~\ref{coro-first-q} has the same system parameters $K$, $M$ and $N$ and the same performance, i.e., the same subpacketization and transmission load, as the ASK scheme in \cite{ASK}. The reason is that the scheme in \cite{ASK} is generated by a symmetric balanced incomplete block design (SBIBD) which is can be constructed by a DS \cite{colbourn2010crc}.

\subsection{Expanding the Concept of NTAP} 
NHSDP includes the well-known non-three-term arithmetic progressions (NTAP) \cite{brown1982density} as a special case. A subset $\mathcal{S}\subseteq\mathbb{Z}_{v}$ is called an NTAP set if any three different elements, say $x$, $y$ and $z$, of $\mathcal{S}$ do not have the property $2z=x+y$. If $\mathcal{S}$ is an NTAP set,  $(\mathbb{Z}_{v},\mathcal{S})$ is an NHSDP with one block.  
So in Theorem~\ref{th:Lagrange}, when $q=3$, we have $b=1$. In this case, \eqref{eq-family-M} can be written as follows.
\begin{align*}
\mathcal{D}=\{\alpha_1+3\alpha_2+\cdots+3^{n-1}\alpha_n\ |\  \alpha_i\in\{-1,1\},i\in[n]\}.
\end{align*} By Theorem~\ref{th-main-2} we know that $\mathcal{D}$ is an NTAP subset of $\mathbb{Z}_{v}$. Then the following result can be obtained.
\begin{lemma}[NTAP set]\rm
\label{le-TAP}
For any positive integer $n$, there exists an NTAP with size of $\rho_1=2^n$ over $\mathbb{Z}_{3^n}$. 
\end{lemma} For any given positive integer $v$ define the maximum size of an NTAP subset of $\mathbb{Z}_{v}$ by $\rho(v)$. The authors in \cite{elsholtz2024improving} showed that $\rho(v) \geq \rho_2=v\cdot 2^{-(2\sqrt{\log_2(24/7)}+o(1))\sqrt{\log_2 v}}$. As far as we know, this is the best lower bound on $\rho(v)$ up to now. 
When $v=3^n$, we have  
\begin{align*}
\frac{\rho_2}{\rho_1}=\frac{v\cdot 2^{-(2\sqrt{\log_2(24/7)}+o(1))\sqrt{\log_2 v}}}{2^n}=\frac{3^n}{2^{(2\sqrt{\log_2(24/7)}+o(1))\sqrt{n\log_2 3}+n}}.
\end{align*}Taking the logarithm of the above ratio, we have 
\begin{align*}
\ln\left(\frac{\rho_2}{\rho_1}\right)=&n\ln 3-((2\sqrt{\log_2(24/7)}+o(1))\sqrt{n\log_2 3}+n)\ln 2\\
=& 1.0986n -0.6931(2.6665+o(1))\cdot 1.5850n^{1/2}+n) \\
=& 1.0986n-  0.6931n-2.9293n^{1/2}-o(1)\cdot1.5850n^{1/2}\\
=&0.4055n-2.9293n^{1/2}-o(1)\cdot1.5850n^{1/2}.
\end{align*}By the above formula, we have $2^n\geq \frac{v\cdot 2^{-(2\sqrt{\log_2(24/7)}+o(1))\sqrt{\log_2 v}}}{2^n}$ when $n\leq \lfloor\left(\frac{2.9293}{0.4055}\right)^2\rfloor=52$,. That is when $n\leq 52$, $\rho_1>\rho_2$ always holds; in other words, by our construction, we propose a tighter lower bound on $\rho(v)$, i.e., $\rho(v)\geq \rho_1$, than the existing one. 

\subsection{Perfect hash family}
\label{sub-PHF} 
 Perfect hash family is a useful combinatorial structure which is generalization of many well-studied objects in combinatorics, cryptography, and coding theory and has many interesting applications, for example, in secure frame-proof codes of fingerprinting problems \cite{stinson2006cryptography}, threshold cryptography \cite{blackburn2023combinatorics,blackburn1996efficient}, covering arrays used in software testing problems \cite{colbourn2006roux}, broadcast encryptions \cite{fiat1994broadcast}, etc. 

Let $\varphi$ be a function from a set $\mathcal{B}$ to a set $\mathcal{C}$. We say that $f$ separates a subset $\mathcal{T}\subseteq \mathcal{B}$ if $\varphi$ is injective when restricted to $\mathcal{T}$. Let $m$, $q$, $t$ be integers such that $m \geq q \geq t \geq 2$.
Suppose $|\mathcal{B}|=m$ and $|\mathcal{C}|=q$. A set $\mathcal{X}$ of functions from $\mathcal{B}$ to $\mathcal{C}$ with $|\mathcal{X}|=r$ is an
$(r;m,q,t)$-perfect hash family if for all $\mathcal{T}\subseteq \mathcal{B}$ with $|\mathcal{T}| = t$, there exists at least one $\varphi\in\mathcal{X}$
such that $\varphi$ separates $\mathcal{T}$. 
We refer to  this as an $(r;m,q,t)$ PHF. An $(r;m,q,t)$ PHF can be depicted as an  $r \times m$ array in which the columns are labeled
by the elements of $\mathcal{B}$, the rows by the functions $\varphi_i\in \mathcal{X}$ and the $(i,j)$-entry of the array
is the value $\varphi_i(j)$. Thus, an $(r;m,q,t)$ PHF is equivalent to an $r \times m$ array with entries from
a set of $q$ symbols such that every $r \times t$ subarray contains at least one row having distinct
symbols.  A perfect hash family is considered optimal if $m$ is as large as possible given $r$, $q$ and $t$.

The authors in \cite{cheng2016bounds,SHF-SG} showed that an NTAP set of size $g$ over $\mathbb{Z}_v$, i.e., $(v,g,1)$ NHSDP, can also be used to construct a $(3;gv,v,3)$ PHF. Using the $(q^2+q+1,q+1,1)$ NHSDP generated by a  $(q^2+q+1,q+1)$ DS for any prime power $q$, we have a can obtain the following new PHF.
\begin{lemma}\rm
\label{le-DS-PHF}
For any prime power $q$, there exists a  $(3;m_1=(q^2+q+1)(q+1),q^2+q+1,3)$ PHF. \hfill $\square$ 
\end{lemma}  
When $q=3$ in Theorem~\ref{th:Lagrange}, we have the following new PHF.
\begin{lemma}\rm
\label{le-TAP-PHF}
For any positive integer $n$, there exists a $(3;m_2=6^n,v=3^n,3)$ PHF. \hfill $\square$ 
\end{lemma} 
By the classic generalized quadrangles, quadrics in PG$(4,p)$, and Hermitian varieties in PG$(4,p^2)$ for any prime power $p$, the authors in \cite{Hara-PHF}  constructed the following PHFs which has maximum $m$ among the existing  deterministic constructions.
\begin{lemma}[\cite{Hara-PHF}]\rm
\label{le-PHF-PG}
For any prime power $p$, there exist a $(3;m_3=p^2(p+1),v_1=p^2,3)$ quadrics PHF and a $(3;m_4=p^5,v_2=p^3,3)$ Hermitian PHF.
\end{lemma} 

By Lemma~\ref{le-DS-PHF} and Lemma~\ref{le-PHF-PG}, we have $m_1\approx m_3$ when $p=q$. Now let us compare the values of $m_2$, $m_3$ and $m_4$ in Lemma~\ref{le-TAP-PHF} and Lemma~\ref{le-PHF-PG}. Let $v=p^2=3^n$ and $v=p^3=3^n$ respectively where $n\geq 2$. We have
\begin{align*}
\frac{m_3}{m_2}=\frac{3^n(3^{n/2}+1)}{6^n}=\left(\frac{3}{4}\right)^{n/2}+\frac{1}{2^n}<1  \ \ \text{and}\ \ \
\frac{m_4}{m_2}=\frac{3^{5n/3}}{6^n}=\frac{3^{2n/3}}{2^n}\approx 1.04^n.
\end{align*}This implies that the PHF proposed in Lemma~\ref{le-TAP-PHF} achieves more columns than the first quadrics PHF and approaches the number of column the Hermitian PHF in \cite{Hara-PHF}.

\section{Conclusion}\label{sec-conclu}
Motivated by the quest to construct coded caching schemes with linear subpacketization and small load, in this work we have introduced a new combinatorial structure called non-half-sum disjoint packing (NHSDP).
This structure is used to construct PDAs with linear subpacketization in the number of users K.  
According to theoretical and numerical comparisons, 
 the proposed scheme achieves lower load than the existing schemes with linear subpacketization; it achieves lower load in some cases than some existing schemes with polynomial subpacketization; it has close loads as some existing schemes with exponential subpacketization in some cases. 
 Furthermore, NHSDP has a close relationship with other classical combinatorial structures such as cyclic difference packing, non-three-term arithmetic progressions, and perfect hash family (PHP). By constructing NHSDPs, we can obtain some new CDPs, NTAPs and PHPs with large cardinalities or columns. 

In this paper, we have focused only on $(v,g,b)$ NHSDPs with $g = 2^n$ for positive integers $n$. In general, we point out as a promising future work to consider also more general NHSDPs with different choice of $g$. 

%
%

\appendices
\section{Proof of Theorem~\ref{th-1}}
\label{sec:proof-th-1}
Assume that  $(\mathbb{Z}_v,\mathfrak{D})$ is a $(v,g,b)$ NHSDP where $\mathfrak{D}=\{\mathcal{D}_1,\mathcal{D}_2,\ldots,\mathcal{D}_b\}$. According to Definition~\ref{def-PDA}, let us compute the number of stars in each column first. For each column $k$ and each each integer $x\in \mathfrak{D}$, there exactly exists a unique integer $f$ satisfying that $k-f=x$. This implies that there are exactly $bg$ integer entries in column $k$, i.e., $v-bg$ star entries in column $k$. So we have $Z=v-bg$. 

Next, for any two different entries say $p_{f_1,k_1}$ and $p_{f_2,k_2}$ satisfying that $ p_{f_1,k_1}=p_{f_2,k_2}=(c,i)$, from \eqref{eq-cons-1} we have 
\begin{align}
\label{eq-equal}d_1=k_1-f_1,\ d_2=k_2-f_2\in \mathcal{D}_i, \ \  \text{and}\ \ k_1+f_1=k_2+f_2. 
\end{align}If $f_1=f_2$ (or $k_1=k_2$) we have $k_1=k_2$ (or $f_1=f_2$) which contradicts our hypothesis that  $p_{f_1,k_1}$ and $p_{f_2,k_2}$ are two different entries. So each integer pair in $\mathbf{P}$ occurs in each row and each column at most once.
Let us consider the case $f_1\neq f_2$ and $k_1\neq k_2$. First from \eqref{eq-equal}, we have $k_1=f_1+d_1, k_2=f_2+d_2$ and\begin{equation}\label{eq-k1k2}
2f_1+d_1=2f_2+d_2,\ \ \ \ \text{i.e.,}\ \ \ \ d_1-d_2=2(f_2-f_1).
\end{equation}
Assume that the entry $p_{f_1,k_2}$ is not star. Then  from \eqref{eq-equal} there exists a integer $i'\in[b]$ and $d_3\in \mathbb{Z}_v$ such that $d_3=k_2-f_1\in\mathcal{D}_{i'}$. This implies that $k_2=d_3+f_1$. Together with $k_2=d_2+f_2$ we have $d_3-d_2=f_2-f_1$. From \eqref{eq-k1k2} we have $2(d_3-d_2)=d_1-d_2$, i.e., $2d_3=d_1+d_2$. This implies that $d_3=\frac{d_1+d_2}{2}\in\mathcal{D}_{i'}$ when $v$ is odd. This contracts the condition of NHSDP that the half-sum of any two different elements in $\mathcal{D}\in \mathfrak{D}$ dose not occur in $\mathfrak{D}$. Similarly we can also show that $p_{f_2,k_1}=*$.

Finally let us compute the number of different integer pairs in $\mathbf{P}$. For any integers $i\in [b]$, $c\in\mathbb{Z}_v$ and for each integer $d\in\mathcal{D}_i$, when $v$ is odd, the following system of equations always has a unique solution:
\begin{align*}
\left\{
\begin{array}{c}
k-f=d;\\
k+f=c. 
\end{array}
\right.
\end{align*}When $d$ runs all the element of $\mathcal{D}_i$, there exactly $g$ unique solutions. This means that the integer pair $(c,i)$ occurs in $\mathbf{P}$ exactly $g$ times and there are exactly $bv$ different integer pairs. Then the proof is completed.

\section{Proof of Lemma~\ref{th-main-2}}
\label{sec:proof of thm2}
By Definition~\ref{def-NHSDP} let us consider the first condition that the intersection of any two difference blocks in $\mathfrak{D}$ is empty. Let ${\bf a}=(a_1,a_2,\ldots,a_n)$, ${\bf a}'=(a'_1,a'_2,\ldots,a'_n)$, ${\bf \alpha}=(\alpha_1,\alpha_2,\ldots,\alpha_n)$, and  ${\bf\alpha}'=(\alpha'_1,\alpha'_2, \ldots,\alpha'_n)$ are four vectors, where ${\bf a}, {\bf a} '\in \mathcal{A}$ and ${\bf\alpha},{\bf\alpha}' \in \{1,-1\}^n$.  From \eqref{eq-family} we have two integers of $\mathfrak{D}$, i.e., 
\begin{align}
\label{eq-x-y}
x=\alpha_1a_1x_1+\cdots+\alpha_na_nx_n\in \mathcal{D}_{{\bf a}}\ \ \text{and}\ \  y=\alpha_1'a_1'x_1+\cdots+\alpha_n'a_n'x_n\in \mathcal{D}_{{\bf a}'}.
\end{align} Since 
$v\geq 2\phi(m_1,m_2,\ldots,m_n)+1=2\sum_{i=1}^{n}f(i)+1$, both $\sum_{i=1}^{n}a_ix_i$ and  $\sum_{i=1}^{n}a'_ix_i$ are less than $\frac{v-1}{2}$. In addition, from \eqref{eq-RF} for each $i\in[n-1]$ we have $x_{i+1}>2\sum_{j=1}^{i}f(j)$ which implies that
\begin{align}
\label{eq-relation}
a_{i+1}x_{i+1}>2\sum_{j=1}^{i}a_jx_j.
\end{align} If $x=y$, we have $\alpha_n=\alpha'_n$ and $a_n=a'_n$. Otherwise, from \eqref{eq-relation} with $i=n-1$, if $a_n\neq a'_n$ we have $x\neq y$. Furthermore, if $\alpha_n\neq \alpha'_n$, without loss of generality, we assume that $\alpha_n<0$ and $\alpha'_n>0$. Then we have $\frac{v-1}{2}< x<v$ and $y\leq \frac{v-1}{2}$, which implies $x\neq y$. This contradicts our hypothesis that $x=y$. So we only need to consider the case $x-\alpha_na_nx_n=y-\alpha'_na'_nx_n$. Similarly we can obtain $a_{n-1}=a'_{n-1}$ and $\alpha_{n-1}=\alpha'_{n-1}$ from \eqref{eq-relation} with $i=n-2$.
Using the aforementioned proof method, we can analogously obtain $a_i=a'_i$ and $\alpha_i=\alpha'_i$, which implies  $\alpha_ia_ix_i=\alpha'_ia'_ix_i$ for each integer $i\in[n]$. Then ${\bf a}={\bf a}'$ and ${\bf \alpha}={\bf \alpha}'$. So each integer in $\mathfrak{D}$ appears exactly once.

Now let us check the property of non-half-sum. For any vector ${\bf a}=(a_1,a_2,\ldots,a_n)\in \mathcal{A}$ and any two different vectors ${\bf \alpha}=(\alpha_1,\alpha_2,
\ldots,\alpha_n)$, ${\bf \alpha}'=(\alpha'_1,\alpha'_2, \ldots,\alpha'_n)\in\{-1,1\}^n$, let us consider the half-sum of integers  
\begin{align}
\label{eq-x-y-h}
x=\alpha_1a_1x_1+\cdots+\alpha_na_nx_n\ \ \text{and}\ \  y=\alpha'_1a_1x_1+\cdots+\alpha'_na_nx_n,
\end{align}i.e., $$\frac{x+y}{2}=\frac{\alpha_1+\alpha'_1}{2}\cdot a_1x_1+\cdots+\frac{\alpha_n+\alpha'_n}{2}\cdot a_nx_n.$$  By our hypothesis ${\bf \alpha}\neq {\bf \alpha}'$ we have $x\neq y$. In addition, since $\alpha_1,\alpha_2\in \{-1,1\}$ it follows that $\frac{\alpha_i+\alpha'_i}{2}\in\{-1,0,1\}$ for each $i\in[n]$. Furthermore there exists at least one integer $i'\in[n]$ such that $\frac{\alpha_i+\alpha'_i}{2}=0$. Otherwise we have ${\bf \alpha}={\bf \alpha}'$ which implies $x=y$. This contradicts our hypothesis $x\neq y$. In the following we will show that $\frac{x+y}{2}$ does not occur in $\mathfrak{D}$. Assume that there exists two vectors ${\bf a}'=(a'_1,a'_2,\ldots,a'_n)\in \mathcal{A}$ and ${\bf \beta}=(\beta_1,\beta_2,\ldots,\beta_n)\in\{-1,1\}^n$ such that  
\begin{align*}
z=a'_1\beta_1x_1+\cdots+a'_n\beta_nx_n=\frac{\alpha_1+\alpha'_1}{2}\cdot a_1x_1+\cdots+\frac{\alpha_n+\alpha'_n}{2}\cdot a_nx_n.
\end{align*}Similar to the proof of the uniqueness of each integer in $\mathfrak{D}$ introduced above, we can get $\beta_i=\frac{\alpha_i+\alpha'_i}{2}$ and $a'_i=a_i$ for every $i\in [n]$. Then there exists at  least one $\beta_{i'}=0$ for some $i'\in[n]$. This contradicts our definition rule given in \eqref{eq-family}, that is, $\beta_{i'} \in \{-1,1\}^n$. So the half-sum of any two different integers in  $\mathfrak{D}$ does not occur in  $\mathfrak{D}$. Then the proof is completed.

\section{Proof of Theorem~\ref{th:Lagrange}}
\label{sec:lagrange}

Clearly this is a convex $n$-th programming problem under the real number domain when $m_1$, $m_2$, $\ldots$, $m_n$ are real variables. 
We can use the Lagrange Multiplier Method to find its optimal solution. Specifically 
by partially differentiating of the function
\begin{align*}
\psi(m_1,\ldots,m_{n})= \prod_{i=1}^{n}m_{i}+\lambda\left(\frac{v-1}{2}-\sum_{i=1}^{n-1}\left(m_i\prod_{j=i+1}^{n}(1+2m_{i})\right)- m_{n}\right).
\end{align*}
With respect to the variable $m_1$, we have  
\begin{align*}
\frac{\partial \psi(m_1,\ldots,m_{n})}{\partial m_1} = \prod_{i=2}^{n} m_i - \lambda \prod_{i=2}^{n}(1 + 2m_i) = 0,
\end{align*}which implies 
\begin{align}
\label{eq-differ-1}
1=\lambda\prod_{i=2}^{n}\left( \frac{1}{m_i}+2\right).
\end{align}Similarly by partially differentiating of $\psi(m_1,\ldots,m_{n})$ with respect to the variable $m_2$, we have 
\begin{align}
\label{eq-differ-2}
\frac{\partial \psi(m_1,\ldots,m_{n})}{\partial m_2} = m_1 \prod_{i=3}^{n} m_i - \lambda \left( 2m_1 \prod_{i=3}^{n}(1 + 2m_i) + \prod_{i=3}^{n}(1 + 2m_i) \right) = 0.
\end{align}Substituting \eqref{eq-differ-1} into \eqref{eq-differ-2} and rearranging it, we have $(2m_2+1)m_1=m_2(2m_1+1)$ which implies 
$m_1=m_2$. Similarly by partially differentiating of $\psi(m_1,\ldots,m_{n})$ with respect to the variable $m_3$, we have 
\begin{align}
\label{eq-differ-3}
\frac{\partial \psi(m_1,\ldots,m_{n})}{\partial m_3}=& m_1m_2 \prod_{i=4}^{n} m_i -\\ 
& \lambda \left( 2m_1(1+2m_2)\prod_{i=4}^{n}(1 + 2m_i) +2m_2\prod_{i=4}^{n}(1 + 2m_i)+\prod_{i=4}^{n}(1 + 2m_i) \right) = 0.\nonumber
\end{align}Substituting \eqref{eq-differ-1} and $m_1=m_2$ into \eqref{eq-differ-3} and rearranging it, we have $(2m_3+1)m_1=m_3(2m_1+1)$ which implies $m_3=m_1$. 

Now let us use  induction to prove that $m_1=m_2=\cdots=m_{n}$. Assume that we have $m_1=m_2=\cdots=m_{k}$ where $3\leq k<n-1$. Now we will show $m_{k+1}=m_1$ by partially differentiating of $\psi(m_1,\ldots,m_{n})$ with respect to the variable $m_{k+1}$, we have 
\begin{align}
\label{eq-differ-k}
\frac{\partial \psi(m_1,\ldots,m_{n})}{\partial m_{k+1}}=&\prod_{i\in [n]\setminus\{k+1\}} m_i-\lambda 
\left(2 \sum_{j=1}^{k-1}\left(m_j\prod_{h\in [j+1:n]\setminus\{k+1\}}(1+2m_{h})\right)+\right.
\\
&\left.2m_k\prod_{i=k+2}^{n}(1 + 2m_i)+
\prod_{i=k+2}^{n}(1 + 2m_i)
\right) = 0.\nonumber
\end{align} Submitting the results $m_1=m_2=\cdots=m_k$ into \eqref{eq-differ-k},  we have
\begin{align}
\label{eq-differ-k'}
m_1^{k}\prod_{j=k+2}^{n} m_j=&\lambda\prod_{i=k+2}^{n}(1 + 2m_i) 
\left(2m_1\sum_{j=0}^{k-1}(1+2m_1)^j +1
\right)\nonumber\\ 
=&(1+2m_1)^{k}\lambda\prod_{i=k+2}^{n}(1 + 2m_i).
\end{align} In addition when $m_1=m_2=\cdots=m_k$, \eqref{eq-differ-1} can be written as follows.
\begin{align}
\label{eq-differ-1'}
\frac{1}{\left(\frac{1}{m_1}+2\right)^{k-1}\left( \frac{1}{m_{k+1}}+2\right)}=\lambda\prod_{i=k+2}^{n}\left(\frac{1}{m_i}+2\right).
\end{align}Submitting the results \eqref{eq-differ-1'} into \eqref{eq-differ-k'},  we have
\begin{align*} 
&m_1^{k}=(1+2m_1)^{k}\frac{1}{\left(\frac{1}{m_1}+2\right)^{k-1}\left( \frac{1}{m_{k+1}}+2\right)}\\
\Longleftrightarrow\ \ \ \ & 
m_{k+1}\left(1+2m_1\right)=m_{k+1}(1+2m_1)\Longleftrightarrow m_{k+1}=m_1.
\end{align*} When $k=n$ we can also show that our statement holds by the above method. 

By the above discussion, we have the optimal solution $m_1=m_2=\cdots=m_{n}$ of the optimization problem \eqref{eq-optimization}. From \eqref{eq-sum-half} and \eqref{eq-optimization} we have
\begin{align*}
\sum_{i=1}^{n-1}\left(m_i\prod_{j=i+1}^{n}(1+2m_{i})\right)+ m_{n}=m_1\sum_{i=0}^{n-1}(1+2m_1)^i=\frac{(1+2m_1)^n-1}{2}\leq \frac{v-1}{2},
\end{align*}which implies $m_1=m_2=\cdots=m_n\leq \frac{\sqrt[n]{v}-1}{2}$. So the total number of blocks of NHSDP is $b\leq (\frac{\sqrt[n]{v}-1}{2})^n$ and the memory ratio is at least $1-\frac{(\sqrt[n]{v}-1)^n}{v}$. Moreover, when  $\sqrt[n]{v}:=q$ for some odd integer $q\geq 3$, \eqref{eq-family} can be written as follows.
\begin{align}
\label{eq-family-M}
\mathcal{D}_{\bf a}=\left\{\sum_{i=1}^{n}\alpha_ia_iq^{i-1}\ \Big|\  \alpha_i\in\{-1,1\},i\in[n]\right\}, \ \ \forall {\bf a}=(a_1,a_2,\ldots,a_n)\in [q]^n.
\end{align}

\bibliographystyle{IEEEtran}
\bibliography{reference}
\end{document}